\documentclass{IEEEtran}
\IEEEoverridecommandlockouts 

\usepackage{graphicx}
\usepackage{subfigure}
\usepackage{color}
\usepackage{epsfig}
\usepackage{amssymb}
\usepackage{amsmath}
\usepackage{amsthm}
\usepackage{latexsym}
\usepackage{setspace}
\usepackage{bbm}
\usepackage{flushend}
\usepackage{float}
\usepackage{dsfont}
\usepackage{caption}
\usepackage[normalem]{ulem}
\usepackage[ruled,vlined]{algorithm2e}

\newcommand{\dP}{\mathrm{P}}
\newcommand{\dQ}{\mathrm{Q}}
\newcommand{\bPP}[1]{{\dP_{#1}}}
\newcommand{\bQQ}[1]{{\dQ_{#1}}}
\newcommand{\bPr}[1]{{\mathbb{P}}\left(#1\right)}

\newcommand{\bP}[2]{\mathrm{P}_{#1}\left({#2}\right)}
\newcommand{\bQ}[2]{\mathrm{Q}_{#1}\left({#2}\right)}

\newcommand{\cA}{{\mathcal A}}

\newcommand{\cE}{{\mathcal E}}
\newcommand{\cF}{{\mathcal F}}

\newcommand{\cH}{{\mathcal H}}

\newcommand{\cJ}{{\mathcal J}}
\newcommand{\cK}{{\mathcal K}}

\newcommand{\cO}{{\mathcal O}}

\newcommand{\cP}{{\mathcal P}}
\newcommand{\cQ}{{\mathcal Q}}

\newcommand{\cT}{{\mathcal T}}

\newcommand{\cV}{{\mathcal V}}

\newcommand{\bx}{\mathbf{x}}
\newcommand{\cX}{{\mathcal X}}
\newcommand{\bX}{{\bf X}}
\newcommand{\by}{\mathbf{y}}
\newcommand{\bY}{{\bf Y}}
\newcommand{\cY}{{\mathcal Y}}
\newcommand{\cZ}{{\mathcal Z}}
\newcommand{\bZ}{{\bf Z}}

\newcommand{\ep}{\epsilon}
\newcommand{\la}{\lambda}

\newcommand{\ie}{{\it i.e.}}
\newcommand{\cf}{{\it cf.}}

\newcommand{\ttlvrn}[2]{\left\| #1 - #2\right\|}

\newtheorem{theorem}{Theorem}
\newtheorem{roughtheorem}{Result}
\newtheorem{proposition}[theorem]{Proposition}

\newtheorem*{corollary*}{Corollary}

\newtheorem{lemma}[theorem]{Lemma}
\newtheorem*{lemma*}{Lemmas}

\theoremstyle{remark}
\newtheorem*{remark*}{Remark}
\newtheorem*{remarks*}{Remarks}
\theoremstyle{definition}
\newtheorem{definition}{Definition}
\newtheorem{example}{Example}
\newtheorem{remark}{Remark}

\allowdisplaybreaks
\newenvironment{protocol}[1][htb]
  {%
   \begin{algorithm}[#1]%
  }{\end{algorithm}}

\renewcommand{\ep}{\varepsilon}
\newcommand{\prot}{\Pi}

\newcommand{\lamin}{\la_{\min}}
\newcommand{\lamax}{\la_{\max}}

\newcommand{\Romn}{\overline{H}({\bf X}\triangle {\bf Y})}
\newcommand{\romn}[2]{h(#1 \triangle #2)}

\newcommand{\ed}[2]{h_{#1}\left(#2\right)}
\newcommand{\edxgy}{\ed {}{X|Y}}

\newcommand{\hash}{\cH}

\begin{document}

\title{Interactive Communication for Data Exchange}

\author{Himanshu Tyagi,~\IEEEmembership{Member,~IEEE}, 
Pramod Viswanath,~\IEEEmembership{Fellow,~IEEE}, 
and
Shun Watanabe,~\IEEEmembership{Member,~IEEE}
\thanks{H. Tyagi is with the Department of Electrical Communication Engineering, Indian
Institute of Science, Bangalore 560012, India. 
Email: htyagi@ece.iisc.ernet.in}
\thanks{P. Viswanath is with the  
Department of Electrical and Computer Engineering,
University of Illinois, Urbana-Champaign, IL 61801, USA.  Email:
\{pramodv\}@illinois.edu
}
\thanks{S. Watanabe is with the Department of Computer and Information Sciences, 
Tokyo University of Agriculture and Technology, Tokyo 184-8588, Japan. 
Email: shunwata@cc.tuat.ac.jp}
}

\maketitle

\renewcommand{\thefootnote}{\arabic{footnote}}
\setcounter{footnote}{0}

\begin{abstract}
Two parties observing correlated data seek to exchange their
data using interactive communication. How many bits
must they communicate? 
We propose a new interactive protocol for data exchange 
which increases the communication size in steps until the task
is done. We also derive a lower bound on the minimum number
of bits that is based on relating the data exchange problem
to the secret key agreement problem.
Our single-shot analysis applies to all discrete random variables
and yields upper and lower bounds of a similar form. In fact,
the bounds are asymptotically tight and lead to
a characterization of 
the optimal rate of communication needed for data exchange 
for a general source sequence such as a mixture of IID random variables
as well as the optimal second-order asymptotic term in the
length of communication needed for data exchange for IID random variables,
when the probability of error is fixed.
This gives a precise characterization of the asymptotic reduction in 
the length of optimal communication due to interaction; in particular, two-sided
Slepian-Wolf compression is strictly suboptimal.
\end{abstract}

\section{Introduction}
Random correlated data $(X,Y)$ is distributed 
between two parties with the first observing $X$
and the second $Y$. What is the optimal communication protocol 
for the two parties to exchange their data? 
We allow
(randomized) interactive communication protocols and 
a nonzero probability of error.
This basic problem was introduced by El Gamal and Orlitsky 
in \cite{OrlEl84} where they presented bounds on 
the average number of bits of communication needed by deterministic 
protocols for data exchange without error\footnote{They also illustrated the advantage of
using randomized protocols when error is allowed}.
When interaction is not allowed,
a simple solution is to apply Slepian-Wolf compression
\cite{SleWol73} for each of the two one-sided data transfer
problems.
The resulting protocol was shown to be of optimal rate, even
in comparison with interactive protocols, when the underlying
observations are {\it independent and identically distributed} 
(IID) by Csisz\'ar and Narayan in \cite{CsiNar04}.
They considered a multiterminal version of this problem, namely the problem of
attaining {\it omniscience}, and established a lower bound
on the rate of communication to show that interaction does
not help in improving the asymptotic rate of communication 
if the probability of error vanishes to $0$.
However, interaction
is known to be beneficial in one-sided data transfer
($cf.$ \cite{Orlitsky90, YanH10, YanHUY08, Dra04}). Can interaction help to 
reduce the communication needed for data exchange, and if so, what 
is the minimum length of interactive 
communication needed for data exchange? 

We address the data exchange problem, 
illustrated in Figure~\ref{f:problem_description},
and provide answers to the questions raised above. 
We provide a new approach for establishing {\it converse bounds} for
problems with interactive communication that relates efficient
communication to secret key agreement and uses the recently
established conditional independence testing bound for the length of a
secret key \cite{TyaWat14}. Furthermore, we propose an {\it
  interactive protocol for data exchange} which matches the
performance of our lower bound in several asymptotic regimes. As a
consequence of the resulting single-shot bounds, we obtain a
characterization of the optimal rate of communication needed for data
exchange for a general sequence $(X_n, Y_n)$ such as a mixture of IID
random variables as well as the optimal second-order asymptotic term
in the length of communication needed for data exchange for the IID
random variables $(X^n, Y^n)$, first instance of such a result in
source coding with interactive communication\footnote{ In a different
  context, recently \cite{AltWag14} showed that the second-order
  asymptotic term in the size of good channel codes can be improved
  using feedback.}.  This in turn leads to a precise characterization
of the gain in asymptotic length of communication due to interaction.

\paragraph*{Related work} The role of interaction in multiparty data
compression has been long recognized. For the 
data exchange problem, this was first studied in \cite{OrlEl84} where
interaction was used to facilitate data exchange by communicating
optimally few bits in a single-shot setup with zero error. In a
different direction, \cite{Dra04, YanH10, YanHUY08} showed that interaction
enables a universal variable-length coding for the Slepian-Wolf
problem (see, also, \cite{FedS02} for a related work on universal
encoding). Furthermore, it was shown in \cite{YanH10}
that the redundancy in variable-length Slepian-Wolf coding with
known distribution can be improved by interaction. In fact,
the first part of our protocol is essentially the same as the one in 
\cite{YanH10} (see, also, \cite{BraRao11}) wherein the length of the
communication is increased in steps until the
second party can decode. In \cite{YanH10}, the step size was chosen to be
$\cO(\sqrt{n})$ for the universal scheme and roughly $\cO(n^{1/4})$
for the known distribution case. We recast this protocol in an information
spectrum framework (in the spirit of \cite{HayTyaWat14ii}) and allow 
for a flexible choice of the step size. By choosing this step size
appropriately, we obtain exact asymptotic results in various regimes.
Specifically, the optimal choice of this step size $\Delta$ is given by the
square root of the essential {\it length of the spectrum} of
$\bPP{X|Y}$, 
$i.e.$, $\Delta = \sqrt{\lamax - \lamin}$ where $\lamax$ and $\lamin$
are large probability upper and lower bounds, respectively, for the
random variable $h(X|Y) = -\log
\bP{X|Y}{X|Y}$. The $\cO(\sqrt{n})$ choice for the universal case of
\cite{YanH10} follows as a special case since for the universal setup with IID source
$h(X^n|Y^n)$ can vary over an interval of length $\cO(n)$. 
Similarly, for a given IID source, $h(X^n|Y^n)$ can essentially vary over an
interval of length $\cO(\sqrt{n})$ for which the choice of $\Delta =
\cO(n^{1/4})$ in \cite{YanH10} is appropriate by our general
principle. While the optimal choice of $\Delta$ (upto the order) was
identified in \cite{YanH10} for special cases, the optimality of this
choice was not shown there. Our main contribution is a converse which
shows that our achieved length of communication is optimal in several
asymptotic regimes. As a by-product, we obtain a precise characterization of gain due to
interaction, one of the few such instances available in the literature. 
Drawing on the techniques introduced in this paper, the much more involved problem of simulation of interactive protocols was addressed in \cite{TyagiVVW16, TyagiVVW17}.   
 
\paragraph*{Organization} The remainder of this paper is organized as follows: We
formally describe the data exchange problem in
Section~\ref{sec:problem}.  Our results are summarized in
Section~\ref{sec:main_results}. Section~\ref{sec:achievability}
contains our single-shot achievability scheme, along with the
necessary prerequisites to describe it, and Section~\ref{sec:converse}
contains our single-shot converse bound.  The strong converse and
second-order asymptotics for the communication length and the optimal
rate of communication for general sources are obtained as a
consequence of single-shot bounds in Section~\ref{s:general_sources}
and \ref{s:strong_converse}, respectively. The final section contains
a discussion of our results and extensions to the error exponent
regime.
\section{Problem formulation} \label{sec:problem}

\begin{figure}[t]
\begin{center}
\includegraphics[scale=0.35]{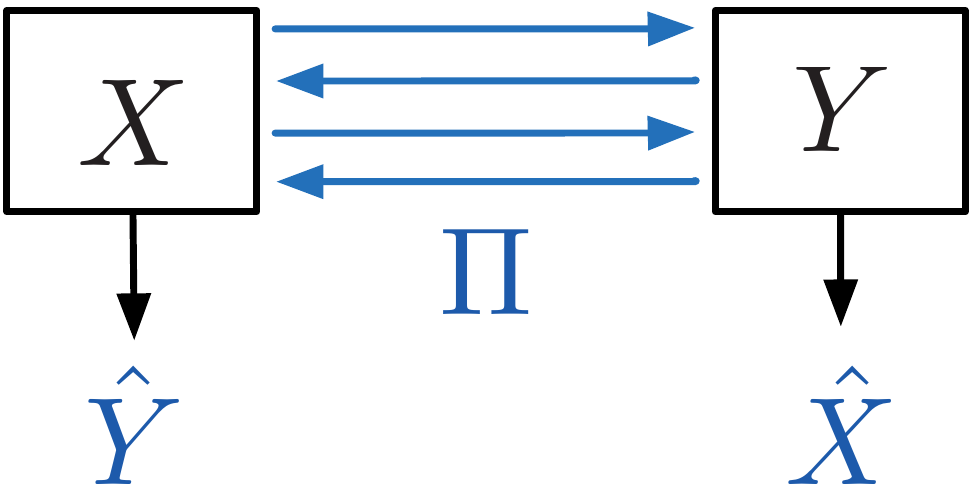}
\caption{The data exchange problem.}
\label{f:problem_description}
\end{center}
\end{figure}

Let the first and the second party, respectively, observe discrete
random variables $X$ and $Y$ taking values in finite sets $\cX$ and
$\cY$. The two parties wish to know each other's observation using
interactive communication over a noiseless (error-free)
channel. The parties have access to local randomness (private coins)
$U_\cX$ and $U_\cY$ and shared randomness (public coins) $U$ such that
the random variables $U_\cX, U_\cY, U$ are finite-valued and mutually
independent and independent jointly of $(X, Y)$.
For simplicity, we resrict to {\it tree-protocols} ($cf.$
\cite{KushilevitzNisan97}).   
A tree-protocol $\pi$ consists of a binary tree, termed the
{\it protocol-tree},  with
the vertices labeled by $1$ or $2$. The protocol starts at
the root and proceeds towards the leaves. When the protocol is at
vertex $v$ with label $i_v\in\{1,2\}$, party $i_v$ communicates a bit
$b_v$ based on its local observations $(X, U_\cX, U)$ for $i_v=1$
or $(Y, U_\cY, U)$ for $i_v =2$. The
protocol proceeds to the left- or right-child of $v$, respectively, if
$b_v$ is $0$ or $1$. The protocol terminates when it reaches a leaf,
at which point each party produces an output based on its local
observations and the bits communicated during the protocol, namely the
transcript $\Pi=\pi(X,Y, U_\cX,U_\cY, U)$. 
Figure~\ref{f:tree_protocols} shows an example of a
protocol tree. 

\begin{figure}[h]
\centering
\includegraphics[scale=0.3]{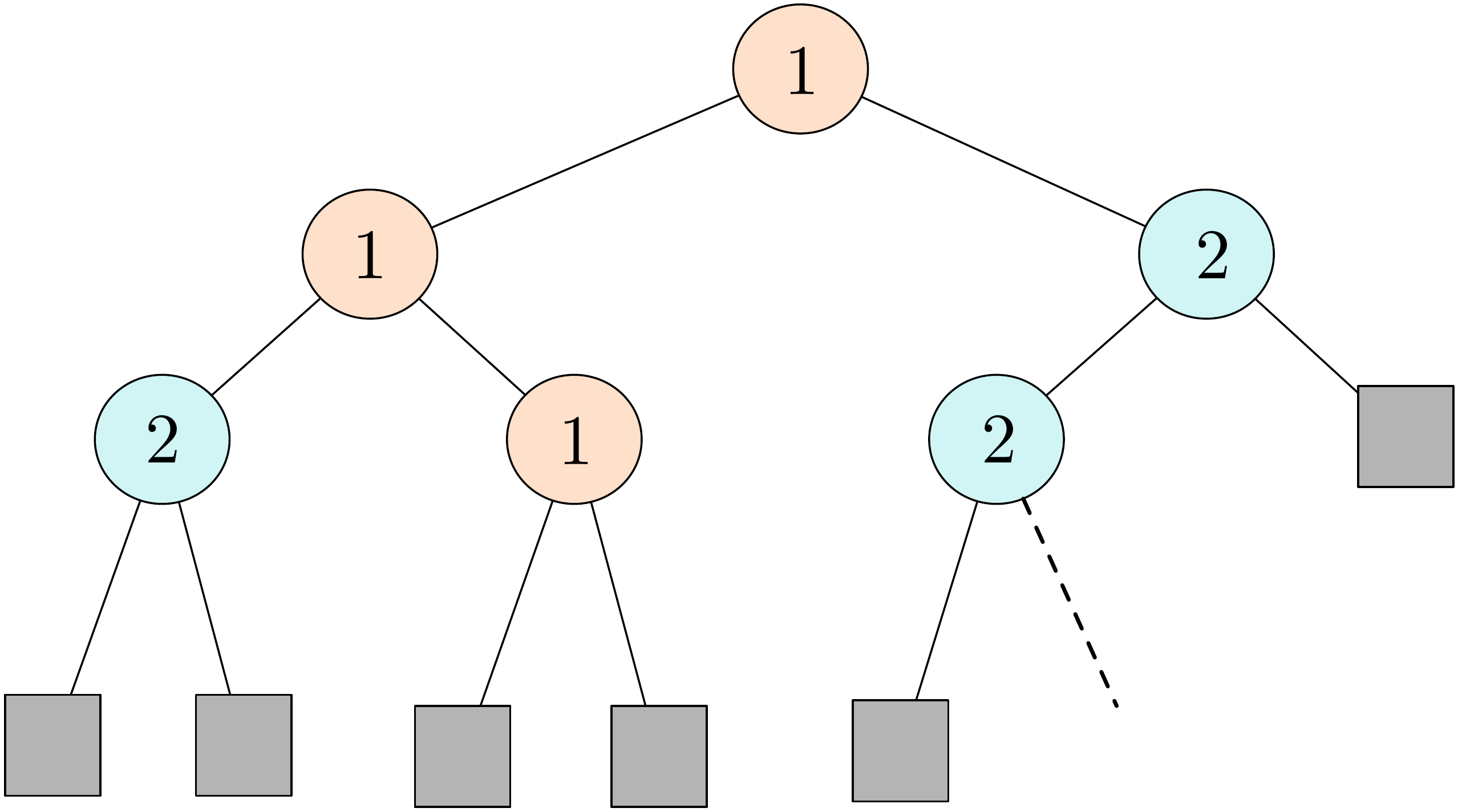}
\caption{A two-party protocol tree.}
\label{f:tree_protocols}
\end{figure}

The {\em length of a protocol} $\pi$, denoted $|\pi|$, is the maximum
accumulated number of bits transmitted in any realization of the
protocol, namely the depth of the protocol tree. 

\begin{definition}
For $0 \le \ep < 1$, a protocol $\pi$ attains $\ep$-{\em data
  exchange} ($\ep$-DE) if there exist functions $\hat{Y}$
and $\hat{X}$ of $(X,\Pi, U_\cX, U)$ and $(Y,\Pi, U_\cX, U)$,
respectively, such that 
\begin{align} 
\bPP{}(\hat{X} = X,~\hat{Y} = Y) \ge 1 - \varepsilon.
\label{e:omn_error_def}
\end{align}
The {\em minimum communication for $\ep$-DE} $L_\ep(X,Y)$ is the
infimum of lengths of protocols\footnote{By derandomizing
  \eqref{e:omn_error_def}, it is easy to see that
local and shared randomness does not help, and deterministic protocols
attain $L_\ep(X,Y)$.} that attain $\ep$-DE, $i.e.$,
$L_\ep(X,Y)$ is the minimum number of bits that must be communicated
by the two parties in order to exchange their observed data with
probability of error less than $\ep$.
\end{definition}
Protocols with $2$ rounds of communication $\Pi_1$ and $\Pi_2$ which
are functions of only $X$ and $Y$, respectively, are termed {\it
  simple protocols}. We denote by
$L_\ep^{\mathrm{s}}(X,Y)$ the minimum communication for $\ep$-DE by a
simple protocol.
\vspace*{-0.1cm}

\section{Summary of results} \label{sec:main_results}
To describe our results, denote
by $h(X) = -\log \bP{X}{X}$ and $h(X|Y) = -\log \bP{X|Y}{X|Y}$,
respectively, the {\it entropy density} of $X$ and the {\it
  conditional entropy density} of $X$ given $Y$.  Also, pivotal in our
results is a quantity we call the {\it sum conditional entropy
  density} of $X$ and $Y$ defined as
\[
\romn XY := h(X|Y) + h(Y|X).
\]

{\bf An interactive data exchange protocol.} Our data exchange
protocol is based on an interactive version of the Slepian-Wolf
protocol where the length of the communication is increased in steps
until the second party decodes the data of the first.  Similar
protocols have been proposed earlier for distributed data compression
in \cite{FedS02, YanH10}, for protocol simulation in \cite{BraRao11},
and for secret key agreement in \cite{HayTyaWat14i,HayTyaWat14ii}.

In order to send $X$ to an observer of $Y$, a single-shot version of
the Slepian-Wolf protocol was proposed in \cite{MiyKan95} (see, also,
\cite[Lemma 7.2.1]{Han03}). Roughly speaking, this protocol simply
hashes $X$ to as many bits as the right most point in the
spectrum\footnote{Spectrum of a distribution $\bPP{X}$ refers, loosely, to the
distribution of the random variable $-\log \bP{X}{X}$.}  of $\bPP{X|Y}$. The main shortcoming
of this protocol for our purpose is that it sends the same number of
bits for every realization of $(X,Y)$. However, we would like to use
as few bits as possible for sending $X$ to party 2 so that the
remaining bits can be used for sending $Y$ to party 1. Note that once
$X$ is recovered by party 2 correctly, it can send $Y$ to Party 1
without error using, say, Shannon-Fano-Elias coding (eg.~see
\cite[Section 5]{CovTho06}); the length of this second communication
is $\lceil h(Y|X) \rceil$ bits. Our protocol accomplishes the first
part above using roughly $h(X|Y)$ bits of communication.

Specifically, in order to send $X$ to $Y$ we use a {\it spectrum
  slicing technique} introduced in \cite{Han03} (see, also,
\cite{HayTyaWat14i, HayTyaWat14ii}).  We divide the support $[\lamin,
  \lamax]$ of spectrum of $\bPP{X|Y}$ into $N$ slices size $\Delta$
each; see Figure~\ref{f:spectrum_slicing} for an illustration.
\begin{figure}[h]
\begin{center}
\includegraphics[scale=0.35]{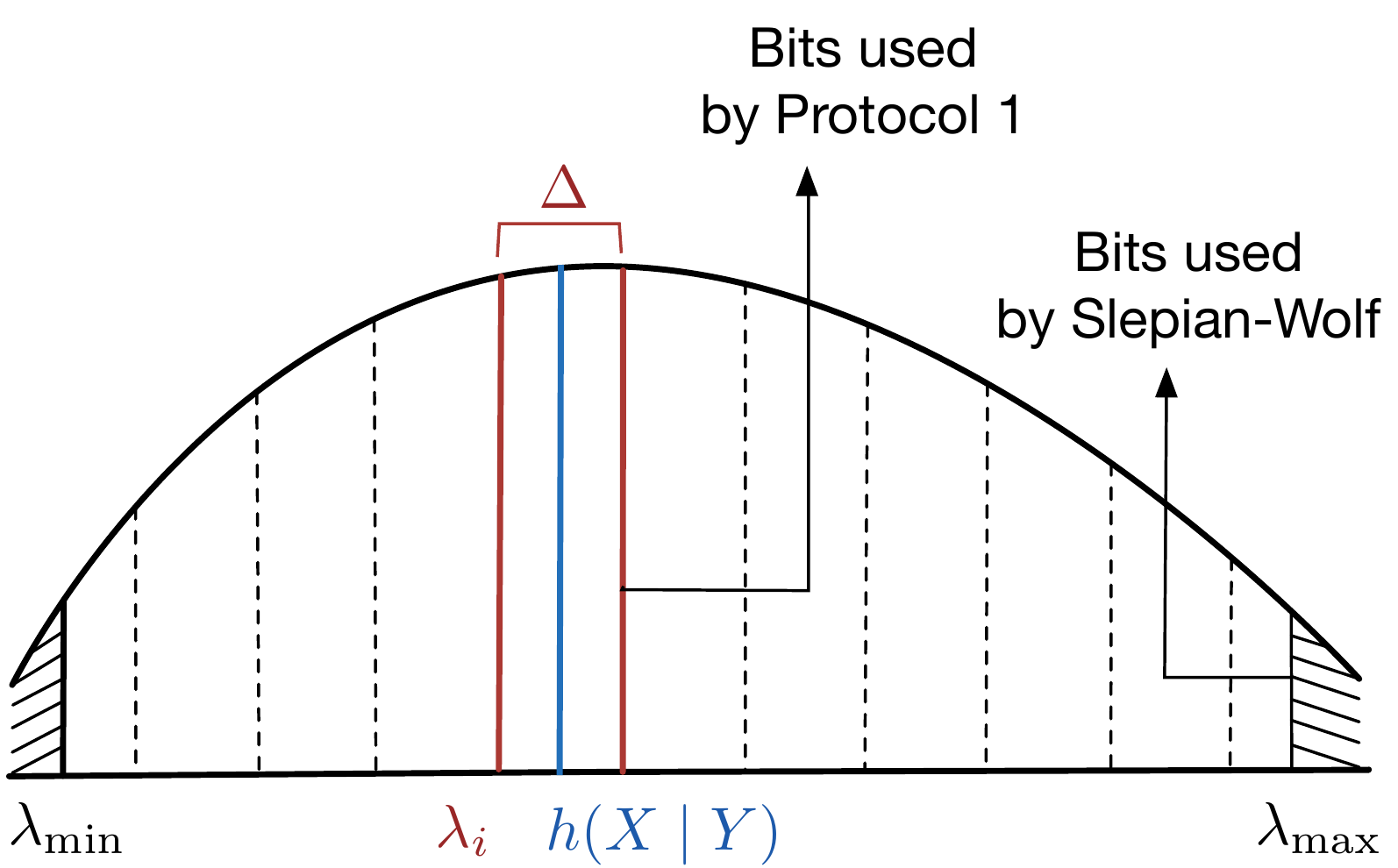}
\caption{Spectrum slicing in
  Protocol~\ref{p:slepian_wolf_interactive}.}
\label{f:spectrum_slicing}
\end{center}
\end{figure}

The protocol begins with the left most slice and party 1 sends
$\lamin+\Delta$ hash bits to party 2. If party 2 can find a unique $x$
that is compatible with the received hash bits and 
$h(x|Y)$ is within the current slice of conditional information spectrum, it sends back an ${\rm
  ACK}$ and the protocol stops. Else, party 2 sends back a ${\rm
  NACK}$ and the protocol now moves to the next round, in which Party
1 sends additional $\Delta$ hash bits. The parties keep on moving to
the next slice until either party 2 sends an ${\rm ACK}$ or all slices
are covered. We will show that this protocol is reliable and
uses no more than $h(X|Y) + \Delta + N$ bits of communication for
each realization of $(X,Y)$. As mentioned above, once party 2 gets
$X$, it sends back $Y$ using $h(Y|X) + 1$ bits, thereby resulting in
an overall communication of $\romn XY+ \Delta +N+1$ bits. In our applications, we
shall choose $N$ and $\Delta$ to be of negligible order in comparison
with the tail bounds for $\romn XY$. Thus, we have the following upper
bound on $L_\ep(X,Y)$. (The statement here is rough; see
Theorem~\ref{t:interactive_data_exchange} below for a precise
version.)
\begin{roughtheorem}[{\bf Rough statement of the single-shot upper bound}] For every $0< \ep <1$,
\[
L_\ep(X,Y) \,\lesssim\, \inf\{\gamma: \bPr{\romn XY >\gamma} \leq
\ep\}.
\]
\end{roughtheorem}

{\bf A converse bound.} Our next result, which is perhaps the main
contribution of this paper, is a lower bound on $L_\ep(X,Y)$.  This
bound is derived by connecting the data exchange problem to the
two-party secret key agreement problem. For an illustration of our
approach in the case of IID random variables $X^n$ and $Y^n$, note
that the optimal rate of a secret key that can be generated is given
by $I(X\wedge Y)$, the mutual information between $X$ and $Y$
\cite{Mau93, AhlCsi93}. Also, using a privacy amplification argument
($cf.$ \cite{BenBraCreMau95, Ren05}), it can be shown that a data
exchange protocol using $nR$ bits can yield roughly $n(H(XY) - R)$
bits of secret key. Therefore, $I(X\wedge Y)$ exceeds $H(XY) - R$,
which further gives 
\[
R \geq H(X| Y) + H(Y| X). 
\]
This connection
between secret key agreement and data exchange was noted first in
\cite{CsiNar04} where it was used for designing an optimal rate secret
key agreement protocol. Our converse proof is, in effect, a
single-shot version of this argument.

Specifically, the ``excess'' randomness generated when the parties
observing $X$ and $Y$ share a communication $\Pi$ can be extracted as
a secret key independent of $\Pi$ using the {\it leftover hash lemma}
\cite{ImpLevLub89, RenWol05}. Thus, denoting by $S_\ep(X,Y)$ the
maximum length of secret key and by $H$ the length of the common
randomness ($cf.$ \cite{AhlCsi93}) generated by the two parties during
the protocol, we get 
\[
H - L_\ep(X,Y) \leq S_\ep(X,Y).
\]

Next, we apply the recently established {\it conditional independence
testing} upper bound for $S_\ep(X,Y)$ \cite{TyaWat14, TyaWat14ii},
which follows by reducing a binary hypothesis testing problem to
secret key agreement. However, the resulting lower bound on
$L_\ep(X,Y)$ is good only when the spectrum of $\bPP{XY}$ is
concentrated. Heuristically, this slack in the lower bound arises
since we are lower bounding the worst-case communication complexity of
the protocol for data exchange -- the resulting lower bound need not
apply for every $(X,Y)$ but only for a few realizations of $(X,Y)$
with probability greater than $\ep$. To remedy this shortcoming, we
once again take recourse to spectrum slicing and show that there
exists a slice of the spectrum of $\bPP{XY}$ where the protocol
requires sufficiently large number of bits;
Figure~\ref{f:spectrum_slicing2} illustrates this approach.  The
resulting lower bound on $L_\ep(X,Y)$ is stated below roughly, and a
precise statement is given in Theorem~\ref{t:converse_general1}.
\vspace{-0.3cm}
\begin{figure}[h]
\begin{center}
\includegraphics[scale=0.35]{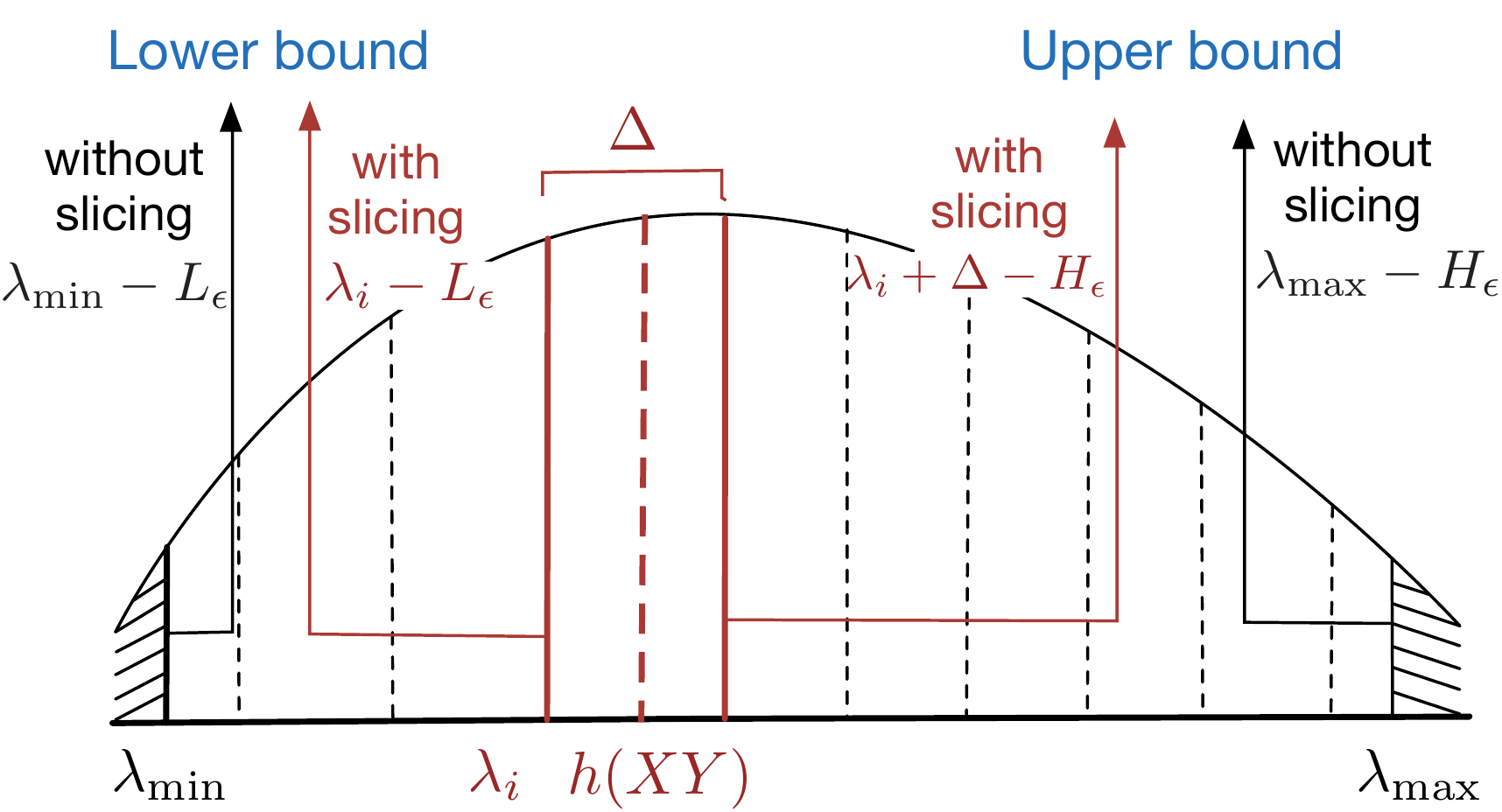}
\caption{Bounds on secret key length leading to the converse.  Here
  $L_\ep$ abbreviates $L_\ep(X,Y)$ and $H_\ep$ denotes the $\ep$-tail of $\romn
  XY$.}
\label{f:spectrum_slicing2}
\end{center}
\end{figure}

\begin{roughtheorem}[{\bf Rough statement of the single-shot lower bound}]
\label{res:lower_bound}
For every $0< \ep <1$,
\[
L_\ep(X,Y) \,\gtrsim\, \inf\{\gamma: \bPr{\romn XY >\gamma} \leq
\ep\}.
\]
\end{roughtheorem}
Note that the upper and the lower bounds for $L_\ep(X,Y)$ in the two
results above appear to be of the same form (upon ignoring a few
error terms). In fact, the displayed term dominates asymptotically
and leads to tight bounds in several asympotitic regimes. Thus, the
imprecise forms above capture the spirit of our bounds. 

\vspace{1em}
{\bf Asymptotic optimality.} The single-shot bounds stated above are
asymptotically tight up to the first order term for any sequence of
random variables $(X_n, Y_n)$, and up to the second order term for a
sequence of IID random variables $(X^n, Y^n)$.

Specifically, consider a general source sequence $(\bX, \bY) = \{(X_n,
Y_n)\}_{n=1}^\infty$.  We are interested in characterizing the minimum
asymptotic rate of communication for asymptotically error-free data
exchange, and seek its comparison with the minimum rate possible using
simple protocols.
\begin{definition}\label{d:R_star}
The minimum rate of communication for data exchange $R^*$ is defined
as
\[
R^*(\bX,\bY) = \inf_{\ep_n}\limsup_{n\rightarrow \infty} \frac 1n
L_{\ep_n}(X_n, Y_n),
\]
where the infimum is over all $\ep_n\rightarrow 0$ as $n \rightarrow
\infty$.  The corresponding minimum rate for simple protocols is
denoted by $R^*_s$.
\end{definition}
Denote by $\Romn$, $\overline{H}(\bX| \bY)$, and
$\overline{H}(\bY|\bX)$, respectively, the $\limsup$ in probability of
random variables $\romn{X_n}{Y_n}$, $h(X_n |Y_n)$, and $h(Y_n|X_n)$.
The quantity $\overline{H}(\bX| \bY)$ is standard in information
spectrum method \cite{HanVer93, Han03} and corresponds to the
asymptotically minimum rate of communication needed to send $X_n$ to
an observer of $Y_n$ \cite{MiyKan95} (see, also, \cite[Lemma
  7.2.1]{Han03}).  Thus, a simple communication protocol of rate
$\overline{H}(\bX| \bY) + \overline{H}(\bY|\bX)$ can be used to
accomplish data exchange. In fact, a standard converse argument can be
used to show the optimality of this rate for simple
communication. Therefore, when we restrict 
ourselves to simple protocols, the asymptotically minimum rate of
communication needed is
\begin{align*}
R^*_s(\bX,\bY) = \overline{H}(\bX| \bY) + \overline{H}(\bY|\bX).
\end{align*}
As an illustration, consider the case when $(X_n, Y_n)$ are generated
by a mixture of two $n$-fold IID distributions
$\mathrm{P}_{X^nY^n}^{(1)}$ and $\mathrm{P}_{X^nY^n}^{(2)}$.  For this
case, the right-side above equals ($cf.$~\cite{Han03})
\begin{align*}
& \max\{H(X^{(1)}\mid Y^{(1)}),H(X^{(2)}\mid Y^{(2)})\} \\
&~~~+ \max\{H(Y^{(1)}\mid X^{(1)}),H(Y^{(2)}\mid X^{(2)})\}.
\end{align*}
Can we improve this rate by using interactive communication?  Using
our single-shot bounds for $L_\ep(X,Y)$, we answer this question in
the affirmative.

\begin{roughtheorem}[{\bf Min rate of communication for data exchange}]
For a sequence of sources $(\bX,\bY) = \{(X_n, Y_n)\}_{n=1}^\infty$,
\[
R^*(\bX,\bY) = \Romn.
\]
\end{roughtheorem}
For the mixture of IID example above,
\begin{align*}
\Romn &= \max\{H(X^{(1)}\mid Y^{(1)}) + H(Y^{(1)}\mid X^{(1)}), \\
&~~~H(X^{(2)}\mid Y^{(2)}) +H(Y^{(2)}\mid X^{(2)})\},
\end{align*}
and therefore, simple protocols are strictly suboptimal in general.
Note that while the standard information spectrum techniques suffice
to prove the converse when we restrict to simple protocols, their
extension to interactive protocols is unclear and our single-shot
converse above is needed.

Turning now to the case of IID random variables, $i.e.$ when $X_n =
X^n = (X_1,..., X_n)$ and $Y_n = Y^n = (Y_1, ..., Y_n)$ are $n$-IID
repetitions of random variables $(X, Y)$. For brevity, denote by
$R^*(X,Y)$ the corresponding minimum rate of communication for data
exchange, and by $H(X\triangle Y)$ and $V$, respectively, the mean and
the variance of $\romn XY$. Earlier, Csisz\'ar and Narayan
\cite{CsiNar04} showed that $R^*(X,Y) = H(X\triangle Y)$. We are
interested in a finer asymptotic analysis than this first order
characterization.

In particular, we are interested in characterizing the asymptotic
behavior of $L_\ep(X^n, Y^n)$ up to the second-order term, for
every fixed $\ep$ in (0,1). We need the following notation:
\[
R^*_\ep(X,Y) = \lim_{n \rightarrow \infty}\frac 1n L_\ep(X^n, Y^n),
\quad 0< \ep < 1.
\]
Note that $R^*(X,Y) = \sup_{\ep\in (0,1)}R^*_\ep(X,Y)$.  Our next
result shows that $R^*_\ep(X,Y)$ does not depend on $\ep$ and
constitutes a {\it strong converse} for the result in \cite{CsiNar04}.
\begin{roughtheorem}[{\bf Strong converse}] \label{result:strong-converse}
For every $0< \ep <1$,
\[
R^*_\ep(X,Y) = H(X\triangle Y).
\]
\end{roughtheorem}
In fact, this result follows from a general result characterizing the
second-order asymptotic term\footnote{Following the pioneering work of
  Strassen \cite{Str62}, study of these second-order terms in coding
  theorems has been revived recently by Hayashi \cite{Hayashi08,
    Hay09} and Polyanskiy, Poor, and Verd\'u~\cite{PolPooVer10}.}.
\begin{roughtheorem}[{\bf Second-order asymptotic behavior}]
For every $0< \ep < 1 $,
\begin{align*}
L_{\ep}\left(X^n, Y^n\right) = nH(X\triangle Y) + \sqrt{n V}
Q^{-1}(\ep) + o(\sqrt{n}),
\end{align*}
where $Q(a)$ is
the tail probability of the standard Gaussian distribution and $V$ is
the variance of the sum conditional entropy density $h(X \triangle
Y)$.
\end{roughtheorem}
While simple protocols are optimal for the first-order term for IID
observations, Example~\ref{ex:second_order_suboptimal} in
Section~\ref{s:strong_converse} exhibits the strict suboptimality of
simple protocols for the second-order term.

\section{A single-shot data exchange protocol} \label{sec:achievability}
We present a single-shot scheme for two parties to exchange random
observations $X$ and $Y$. As a preparation for our protocol, we
consider the restricted problem where only the second party observing
$Y$ seeks to know the observation $X$ of the first party. This basic
problem was introduced in the seminal work of Slepian and Wolf
\cite{SleWol73} for the case where the underlying data is IID where a
scheme with optimal rate was given.  A single-shot version of the
Slepian-Wolf scheme was given in \cite{MiyKan95} (see, also,
\cite[Lemma 7.2.1]{Han03}).  which we describe below.

Using the standard ``random binning'' and ``typical set'' decoding argument,
it follows that there exists an $l$-bit communication $\Pi_1 =
\Pi_1(X)$ and a function $\hat X$ of $(\Pi_1, Y)$ such that
\begin{align}
\bPr{X \neq \hat{X}} \le \bPr{h(X| Y) > l-\eta} + 2^{-\eta}.
\label{e:SW_error_bound}
\end{align}

In essence, the result of \cite{MiyKan95} shows that we can send $X$
to Party 2 with a probability of error less than $\ep$ using roughly as
many bits as the $\ep$-tail of $h(X|Y)$. However, the proposed scheme
uses the same number of bits for every realization of $(X,Y)$.  In
contrast, we present an interactive scheme that achieves the same goal
and uses roughly $h(X|Y)$ bits when the underlying observations are
$(X,Y)$

While the bound in \eqref{e:SW_error_bound} can be used to establish
the asymptotic rate optimality of the Slepian-Wolf scheme even for
general sources, the number of bits communicated can be reduced for
specific realizations of $X,Y$. This improvement is achieved using an
interactive protocol with an ${\rm ACK-NACK}$ feedback which halts as
soon as the second party decodes first's observation; this protocol is
described in the next subsection.  A similar scheme was introduced by
Feder and Shulman in \cite{FedS02}, a variant of which was shown to be
of least {\it average-case complexity} for stationary sources by Yang
and He in \cite{YanH10}, requiring $H(X\mid Y)$ bits on average.
Another variant of this scheme has been used recently in
\cite{HayTyaWat14ii} to generate secret keys of optimal asymptotic length
upto the second-order term.

\subsection{Interactive Slepian-Wolf Compression Protocol}
We begin with an interactive scheme for sending $X$ to an observer of
$Y$, which hashes (bins) $X$ into a few values as in the scheme of
\cite{MiyKan95}, but unlike that scheme, increases the hash-size
gradually, starting with $\lambda_1 = \lambda_{\min}$ and increasing
the size $\Delta$-bits at a time until either $X$ is recovered or
$\lambda_{\max}$ bits have been sent. After each transmission, Party 2
sends either an $\mathrm{ACK}$-$\mathrm{NACK}$ feedback signal; the
protocol stops when an $\mathrm{ACK}$ symbol is received.

As mentioned in the introduction, we rely on spectrum slicing. Our
protocol focuses on the ``essential spectrum'' of $h(X|Y)$, $i.e.$,
those values of $(X,Y)$ for which $h(X|Y) \in (\lamin, \lamax)$.  For
$\lambda_{\min}, \lambda_{\max}, \Delta > 0$ with $\lambda_{\max} >
\lambda_{\min}$, let
\begin{align} \label{eq:number-of-round}
N = \frac{\lambda_{\max} - \lambda_{\min}}{\Delta},
\end{align}
and
\begin{align}
\lambda_i = \lambda_{\min} + (i-1) \Delta,~~1 \le i \le N.
\label{e:lambda_i}
\end{align}
Further, let
\begin{align}
\cT_0 = \Big\{ (x,y) : h_{\bPP{X|Y}}(x|y) \ge \lambda_{\max} \mbox{ or
} h_{\bPP{X|Y}}(x|y) < \lambda_{\min} \Big\},
\label{e:typical}
\end{align}
and for $1 \le i \le N$, let $\cT_i$ denote the $i$th slice of the
spectrum given by
\begin{align}
\cT_i = \Big\{ (x,y) : \lambda_i \le h_{\bPP{X|Y}}(x|y) < \lambda_i +
\Delta \Big\}.
\label{e:slice_i}
\end{align}
Note that $\cT_0$ corresponds to
the complement of the ``typical set.''  Finally, let $\cH_l(\cX)$ denote
the set of all mappings $h:\cX \to \{0,1\}^l$.

Our protocol for transmitting $X$ to an observer of $Y$ is described
in Protocol \ref{p:slepian_wolf_interactive}. The lemma below bounds
the probability of error for Protocol \ref{p:slepian_wolf_interactive}
when $(x,y)\in \cT_i$, $1\le i \le N$.
\begin{protocol}[h]
\caption{Interactive Slepian-Wolf compression}
\label{p:slepian_wolf_interactive}
\KwIn{Observations $X$ and $Y$, uniform public randomness $U$, and
  parameters $l, \Delta$} \KwOut{Estimate $\hat X$ of $X$ at party 2}
Both parties use $U$ to select $h_1$ uniformly from
$\hash_{l}(\cX)$\\ Party 1 sends $\prot_1 = h_1(X)$\\ \eIf{Party 2
  finds a unique $x\in\cT_1$ with hash value $h_1(x) = \prot_1$} { set
  $\hat X = x$\newline send back $\prot_{2} = \text {ACK}$ } { send
  back $\prot_{2} = \text{NACK}$ } \While{ $2\le i\le N$ and party 2
  did not send an ACK} { Both parties use $U$ to select $h_i$
  uniformly from $\hash_{\Delta}(\cX)$, independent of $h_1, ...,
  h_{i-1}$\\ Party 1 sends $\prot_{2i-1} = h_i(X)$\\ \eIf{Party 2
    finds a unique $x\in\cT_i$ with hash value $h_j(x) =
    \prot_{2j-1},\, \forall\, 1\le j\le i$} { set $\hat X = x$\newline
    send back $\prot_{2i} = \text{ACK}$ } { \eIf{More than one such
      $x$ found} { protocol declares an error } { send back
      $\prot_{2i} = \text{NACK}$ } } Reset $i \rightarrow i+1$ }
\If{No $\hat X$ found at party 2} { Protocol declares an error }
\end{protocol}
\begin{theorem}[{\bf Interactive Slepian-Wolf}]
\label{t:slepian_wolf_interactive}
Protocol~\ref{p:slepian_wolf_interactive} with $l = \lamin+\Delta +
\eta$ sends at most $(\edxgy + \Delta + N+\eta)$ bits when the
observations are $(X,Y) \notin \cT_0$ and has probability of error
less than
\[
\bPr{\hat X \neq X} \leq \bP{XY}{\cT_0} + N2^{-\eta}.
\]
\end{theorem}
Note that when $\cT_0$ is chosen to be of small probability,
Protocol~\ref{p:slepian_wolf_interactive} sends essentially the same
number of bits in the worst-case as the Slepian-Wolf protocol.
\subsection{Interactive protocol for data exchange}

Returning to the data exchange problem, our protocol for data exchange
builds upon Protocol~\ref{p:slepian_wolf_interactive} and uses it to
first transmit $X$ to the second party (observing $Y$).  Once Party 2
has recovered $X$ correctly, it sends $Y$ to Party 1 without error
using, say, Shannon-Fano-Elias coding (eg.~see \cite[Section
  5]{CovTho06}); the length of this second communication is $\lceil
h(Y|X) \rceil$ bits. When the accumulated number of bits communicated
in the protocol exceeds a prescribed length $l_{\max}$, the parties
abort the protocol and declare an error.\footnote{Alternatively, we
  can use the (noninteractive) Slepian-Wolf coding by setting the size
  of hash as $l_{\max} - (h(X|Y) + \Delta+N+\eta)$.}  Using
Theorem~\ref{t:slepian_wolf_interactive}, the probability of error of
the combined protocol is bounded above as follows.
\begin{theorem}[{\bf Interactive data exchange protocol}]\label{t:interactive_data_exchange} 
Given $\lambda_{\min},\lambda_{\max}, \Delta,\eta > 0$ and for $N$ in
\eqref{eq:number-of-round}, there exists a protocol for data exchange
of length $l_{\max}$ such that
\begin{align*}
\lefteqn{ \bPr{X \neq \hat{X} \mbox{ or } Y \neq \hat{Y}} } \\
&\le \bPr{\romn XY +
    \Delta + N +\eta+ 1 > l_{\max}} \\
    &~~~ + \bP{XY}{\cT_0} + N 2^{-\eta}.
\end{align*}
\end{theorem} 
Thus, we attain $\ep$-DE using a protocol of length 
\[
l_{\max} = \lambda_\ep+ \Delta+N + \eta+1, 
\]
where $\lambda_{\ep}$ is the $\ep$-tail of
$\romn XY$.  Note that using the noninteractive Slepian-Wolf protocol
on both sides will require roughly as many bits of communication as
the sum of $\ep$-tails of $h(X|Y)$ and $h(Y|X)$, which, in general, is
more than the $\ep$-tail of $h(X|Y) + h(Y|X)$.
\vspace*{-0.1cm}

\subsection{Proof of Theorem~\ref{t:slepian_wolf_interactive}} 
The theorem follows as a corollary of the following observation.
\begin{lemma}[{\bf Performance of Protocol~\ref{p:slepian_wolf_interactive}}]
\label{l:slepian_wolf_interactive}
For $(x,y)\in \cT_i$, $1\le i \le N$, denoting by $\hat X = \hat
X(x,y)$ the estimate of $x$ at Party 2 at the end of the protocol
(with the convention that $\hat X = \emptyset$ if an error is
declared), Protocol~\ref{p:slepian_wolf_interactive} sends at most
$(l+(i-1)\Delta + i)$ bits and has probability of error bounded above
as follows:
\[
\bPr{\hat X \neq x \mid X=x, Y=y} \leq i2^{\la_{\min}+\Delta - l}.
\]
\end{lemma}
{\it Proof.}  Since $(x,y)\in \cT_i$, an error occurs if there exists
a $\hat{x}\neq x$ such that $(\hat x,y)\in \cT_j$ and $\prot_{2k-1} =
h_{2k-1}(\hat x)$ for $1 \leq k \leq j$ for some $j\leq i$. Therefore,
the probability of error is bounded above as
\begin{align*}
& \bPr{\hat X \neq x \mid X=x, Y=y} \\ &\leq \sum_{j=1}^{i}\sum_{\hat
    x\neq x} \bPr{h_{2k-1}(x) = h_{2k-1}(\hat x),\, \forall\, 1\leq k
    \leq j} \\
&~~~\times    \mathbbm{1}\big((\hat x, y)\in \cT_j\big) \\
    &\leq
  \sum_{j=1}^{i}\sum_{\hat x\neq x}
  \frac{1}{2^{l+(j-1)\Delta}}\mathbbm{1}\big((\hat x, y)\in \cT_j\big) \\
  &= \sum_{j=1}^{i}\sum_{\hat x\neq x}\frac{1}{2^{l+(j-1)\Delta}}
  |\{\hat x \mid (\hat x, y)\in \cT_j\}| \\&\leq i2^{\lamin - l
    +\Delta},
\end{align*}
where we have used the fact that $\log |\{\hat x \mid (\hat x, y)\in
\cT_j\}| \leq \la_j+\Delta$.  Note that the protocol sends $l$ bits in
the first transmission, and $\Delta$ bits and $1$-bit feedback in
every subsequent transmission.  Therefore, no more than
$(l+(i-1)\Delta + i)$ bits are sent.  \qed

\section{Converse bound} \label{sec:converse}
Our converse bound, while heuristically simple, is technically
involved.  We first state the formal statement and provide the high level
ideas underlying the proof; the formal proof will be provided later.

Our converse proof, too, relies on spectrum slicing to find the part
of the spectrum of $\bPP{XY}$ where the protocol communicates large
number of bits. As in the achievability part, we shall focus on the
``essential spectrum'' of $h(XY)$.

Given $\lamax$, $\lamin$, and $\Delta>0$, let $N$ be as in
\eqref{eq:number-of-round} and the set $\cT_0$ be as in
\eqref{e:typical}, with $h_{\bPP{X|Y}}(x|y)$ replaced by
$h_{\bPP{XY}}(xy)$ in those definitions.

\begin{theorem}\label{t:converse_general1}
For $0\le \ep <1$, $0< \eta < 1-\ep$, and parameters $\Delta, N$ as
above, the following lower bound on $L_\ep(X,Y)$ holds for every
$\gamma>0$:
\begin{align*}
L_\ep(X,Y) &\ge \gamma + 3\log\left(\bPP \gamma- \ep - \bP{XY}{\cT_0} -
\frac 1N\right)_+ \\
&~~~+\log(1-2\eta)-\Delta -6\log N - 4\log\frac 1
      {\eta}-1,
\end{align*}
where $\bPP \gamma := \bP{XY}{\romn{X}{Y} > \gamma}$.
\end{theorem} 
Thus, a protocol attaining $\ep$-DE must communicate roughly as many
bits as $\ep$-tail of $\romn XY$.

The main idea is to relate data exchange to secret key agreement,
which is done in the following two steps:
\begin{enumerate}
\item Given a protocol $\pi$ for $\ep$-DE of length $l$, use the
  leftover hash lemma to extract an $\ep$-secret key of length roughly
  $\lamin - l$.

\item The length of the secret key that has been generated is bounded
  above by $S_\ep(X,Y)$, the maximum possible length of an
  $\ep$-secret key. Use the conditional independence testing bound in
  \cite{TyaWat14, TyaWat14ii} to further upper bound $S_\ep(X,Y)$,
  thereby obtaining a lower bound for $l$.
\end{enumerate}
This approach leads to a loss of $\lamax- \lamin$, the length of the
spectrum of $\bPP{XY}$.  However, since we are lower bounding the
worse-case communication complexity, we can divide the spectrum into
small slices of length $\Delta$, and show that there is a slice where
the communication is high enough by applying the steps above to the
conditional distribution given that $(X,Y)$ lie in a given slice. This
reduces the loss from $\lamax-\lamin$ to $\Delta$.

\subsection{Review of two party secret key agreement}\label{s:secret_keys}
Consider two parties with the first and the second party,
respectively, observing the random variable $X$ and $Y$.  Using an
interactive protocol $\Pi$ and their local observations, the parties
agree on a secret key. A random variable $K$ constitutes a secret key
if the two parties form estimates that agree with $K$ with probability
close to $1$ and $K$ is concealed, in effect, from an eavesdropper
with access to communication $\Pi$.  Formally, let $K_x$ and $K_y$,
respectively, be randomized functions of $(X, \Pi)$ and $(Y,
\Pi)$. Such random variables $K_x$ and $K_y$ with common range $\cK$
constitute an {\it$\ep$-secret key} ($\ep$-SK) if the following
condition is satisfied:
\begin{eqnarray}
\frac{1}{2}\left\| \bPP{K_xK_y\Pi} -
\mathrm{P}_{\mathtt{unif}}^{(2)}\times \bPP{\Pi}\right\| &\leq \ep,
\nonumber
\end{eqnarray}
where
\begin{eqnarray*}
\mathrm{P}_{\mathtt{unif}}^{(2)}\left(k_x, k_y\right) =
\frac{\mathbbm{1}(k_x= k_y)}{|\cK|},
\end{eqnarray*}
and $\| \cdot \|$ is the variational distance.  The condition above
ensures both reliable {\it recovery}, requiring $\bPr{K_x \neq K_y}$ to be
small, and information theoretic {\it secrecy}, requiring the
distribution of $K_x$ (or $K_y$) to be almost independent of the
communication $\Pi$ and to be almost uniform.
See \cite{TyaWat14} for a discussion on connections between the
combined condition above and the usual separate conditions for
recovery and secrecy.

\begin{definition}
Given $0\le \ep < 1$, the supremum over lengths $\log|\cK|$ of an
$\ep$-SK is denoted by $S_\ep(X, Y)$.
\end{definition}

A key tool for generating secret keys is the {\it leftover hash lemma}
\cite{ImpLevLub89, RenWol05} which, given a random variable $X$ and an
$l$-bit eavesdropper's observation $Z$, allows us to extract roughly
$H_{\min}(\bPP X) - l$ bits of uniform bits, independent of $Z$.  Here
$H_{\min}$ denotes the {\it min-entropy} and is given by
\[
H_{\min}\left(\bPP X\right) = \sup_x \log \frac{1}{\bP X x}.
\]
Formally, let $\cF$ be a {\it $2$-universal family} of mappings $f:
\cX\rightarrow \cK$, $i.e.$, for each $x'\neq x$, the family $\cF$
satisfies
\[
\frac{1}{|\cF|} \sum_{f\in \cF} \mathbbm{1}(f(x) = f(x')) \leq
\frac{1}{|\cK|}.
\]

\begin{lemma}[{\bf Leftover Hash}]\label{l:leftover_hash} Consider
  random variables $X$ and $Z$ taking values in a
countable set $\cX$ and a finite set $\cZ$, respectively.  Let $S$ be a
random seed such that $f_S$ is uniformly distributed over a
$2$-universal family $\cF$.  Then, for $K = f_S(X)$
\begin{align*}
\ttlvrn{\bPP{KZS}}{\bPP{\mathtt{unif}}\bPP{ZS}} \leq
\sqrt{|\cK||\cZ|2^{- H_{\min} \left(\bPP{X}\right)}},
\end{align*}
where $\bPP{\mathtt{unif}}$ is the uniform distribution on $\cK$.
\end{lemma}
The version above is a straightforward modification of the leftover
hash lemma in, for instance, \cite{Ren05} and can be derived in a
similar manner (see Appendix B of \cite{HayTyaWat14ii}).

Next, we recall the {\it conditional independence testing} upper bound
on $S_{\ep}(X, Y)$, which was established in \cite{TyaWat14,
  TyaWat14ii}. In fact, the general upper bound in \cite{TyaWat14,
  TyaWat14ii} is a single-shot upper bound on the secret key length
for a multiparty secret key agreement problem with side information at
the eavesdropper. Below, we recall a specialization of the general
result for the two party case with no side information at the
eavesdropper. In order to state the result, we need the following
concept from binary hypothesis testing.

Consider a binary hypothesis testing problem with null hypothesis
$\mathrm{P}$ and alternative hypothesis $\mathrm{Q}$, where
$\mathrm{P}$ and $\mathrm{Q}$ are distributions on the same alphabet
${\cal V}$. Upon observing a value $v\in \cV$, the observer needs to
decide if the value was generated by the distribution $\bPP{}$ or the
distribution $\mathrm{Q}$. To this end, the observer applies a
stochastic test $\mathrm{T}$, which is a conditional distribution on
$\{0,1\}$ given an observation $v\in \cV$. When $v\in \cV$ is
observed, the test $\mathrm{T}$ chooses the null hypothesis with
probability $\mathrm{T}(0|v)$ and the alternative hypothesis with
probability $T(1|v) = 1 - T(0|v)$.  For $0\leq \ep<1$, denote by
$\beta_\ep(\mathrm{P},\mathrm{Q})$ the infimum of the probability of
error of type II given that the probability of error of type I is less
than $\ep$, \ie,
\begin{eqnarray}
\beta_\ep(\mathrm{P},\mathrm{Q}) := \inf_{\mathrm{T}\, :\,
  \mathrm{P}[\mathrm{T}] \ge 1 - \ep} \mathrm{Q}[\mathrm{T}],
\label{e:beta-epsilon}
\end{eqnarray}
where
\begin{eqnarray*}
\mathrm{P}[\mathrm{T}] &=& \sum_v \mathrm{P}(v) \mathrm{T}(0|v),
\\ \mathrm{Q}[\mathrm{T}] &=& \sum_v \mathrm{Q}(v) \mathrm{T}(0|v).
\end{eqnarray*}

The following upper bound for $S_\ep(X,Y)$ was established in
\cite{TyaWat14, TyaWat14ii}.
\begin{theorem}[{\bf Conditional independence testing bound}] \label{theorem:one-shot-converse-source-model}
Given $0\leq \ep <1$, $0<\eta<1-\ep$, the following bound holds:
\begin{eqnarray}
S_{\ep}\left(X, Y\right) \le -\log
\beta_{\ep+\eta}\big(\bPP{XY},\mathrm{Q}_{X}\mathrm{Q}_{Y}\big) + 2
\log(1/\eta), \nonumber
\end{eqnarray}
for all distributions $\bQQ X$ and $\bQQ Y$ on on $\cX$ and $\cY$,
respectively.
\end{theorem}
We close by noting a further upper bound for $\beta_\ep(\bPP{},
\bQQ{})$, which is easy to derive ($cf.$~\cite{Pol10}).
\begin{lemma}\label{l:bound_beta_epsilon}
For every $0\leq \ep \leq 1$ and $\lambda$,
\[
-\log \beta_\ep(\bPP{}, \bQQ{}) \leq \lambda -
\log\left(\mathrm{P}\left(\log\frac{ \bP {} X}{\bQ{} X} <
\lambda\right) - \ep\right)_+,
\]
where $(x)_+ = \max\{0,x\}$.
\end{lemma}

\subsection{Converse bound for almost uniform distribution}

First, we consider a converse bound under the almost uniformity
assumption.  Suppose that there exist $\lamin$ and $\lamax$ such that
\begin{align}
&\lamin \le -\log \bPP {XY}(x,y) \le \lamax,
\nonumber
\\
&\hspace{4cm} \forall (x,y) \in
\mathrm{supp}(\bPP{XY}),
\label{e:uniformity_assumption}
\end{align}
where $\mathrm{supp}(\bPP{XY})$ denotes the support of $\bPP{XY}$. We
call such a distribution $\bPP{XY}$ an almost uniform distribution
with margin $\Delta=(\lamax- \lamin)$.

\begin{theorem}\label{t:converse_AU}
Let $\bPP{XY}$ be almost uniform with margin $\Delta$. 
Given $0\le \ep <1$, for every $0< \eta< 1-\ep$, and all distributions
$\bQQ X$ and $\bQQ Y$, it holds that
\begin{align*}
\lefteqn{ L_\ep(X,Y) } \\ 
&\ge \gamma + \log\left(\bPr{-\log \frac{\bP
    {XY}{X,Y}^2}{\bQ{X}{X}\bQ{Y}{Y}} \geq \gamma}- \ep -
2\eta\right)_+ \\
&~~~ -\Delta - 4\log\frac 1 {\eta}-1.
\end{align*}
\end{theorem}
\begin{remark}
If $\Delta\approx 0$ (the almost uniform case), the bound above yields
Result~\ref{res:lower_bound} upon choosing $\bQQ X = \bPP X$ and $\bQQ
Y = \bPP Y$. 
\end{remark}
{\it Proof.} Given a protocol $\pi$ of length $l$ that attains
$\ep$-DE, 
 using Lemma~\ref{l:leftover_hash} we can generate an
$(\ep+\eta)$-SK that is almost independent of $\Pi$ and takes values in $\cK$ with
\[
\log |\cK| \geq \lamin - l - 2\log(1/\eta)-1.
\]
Also, by Theorem~\ref{theorem:one-shot-converse-source-model}
\[
\log|\cK| \leq -\log \beta_{\ep+2\eta}(\bPP {XY}, \bQQ{X}\bQQ{Y})
+2\log(1/\eta),
\]
which along with the inequality above and
Lemma~\ref{l:bound_beta_epsilon} yields
\begin{align*}
l &\geq \lamin + \log\left(\bPr{\log\frac{ \bP {XY}{X, Y}}{\bQ{X}
    {X}\bQ{Y}{Y}} < \lambda} - \ep-2\eta\right)_+ \\
    &~~~ - \lambda
-4\log(1/\eta)-1.
\end{align*}
The claimed bound follows upon choosing $\lambda = \lamax-\gamma$ and
using assumption \eqref{e:uniformity_assumption}.  \qed

\subsection{Converse bound for all distributions}
The shortcoming of Theorem~\ref{t:converse_AU} is the $\Delta$-loss,
which is negligible only if $\lamax\approx \lamin$. To circumvent
this loss, we divide the spectrum of $\bPP {XY}$ into slices such that,
conditioned on any slice, the distribution is almost uniform with a
small margin $\Delta$. To lower bound the
worst-case communication complexity of a given protocol, we identify a
particular slice where appropriately many bits are
communicated; the required slice is selected using
Lemma~\ref{l:good_index} below.

Given $\lamax$, $\lamin$, and $\Delta>0$, let $N$ be as in
\eqref{eq:number-of-round}, $\cT_0$ be as in \eqref{e:typical}, and
$\lambda_i$ and $\cT_i$, too, be as defined there, with
$h_{\bPP{X|Y}}(x|y)$ replaced by $h_{\bPP{XY}}(xy)$ in those
definitions.  Let random variable $J$ take the value $j$ when $\{(X,Y)
\in \cT_j\}$.  For a protocol $\Pi$ attaining $\ep$-DE, denote
\begin{align}
\cE_{\mathtt{correct}} &:= \{X= \hat X, Y = \hat Y\}, 
\nonumber
\\ 
\cE_\gamma
&:= \{\romn XY\ge \gamma\}, 
\label{e:E_gamma_def}
\\ 
\cE_j &:= \cE_{\mathtt{correct}} \cap
\cT_0^c\cap \cE_\gamma \cap \{J=j\},\quad 1\leq j \leq N, 
\nonumber
\\ \bPP
\gamma &:= \bP {XY}{\cE_\gamma}.
\nonumber
\end{align}
\begin{lemma}\label{l:good_index}
There exists an index $1\leq j\leq N$ such that $\bP J j> 1/N^2$ and
\[
\bP {XY\mid J}{\cE_j\mid j}\ge \left(\bPP \gamma -\ep - \bP{XY}{\cT_0}
- \frac 1N\right).
\]
\end{lemma}
{\it Proof.} Let $\cJ_1$ be the set of indices $1\leq j \leq N$ such
that $\bP J j >1/N^2$, and let $\cJ_2 = \{1, ..., N\}\setminus
\cJ_1$. Note that $\bP J{\cJ_2} \leq 1/N$. Therefore,
\begin{align*}
\bPP \gamma - \ep - \bP{XY}{\cT_0} &\le \bPr
     {\cE_{\mathtt{correct}}\cap \cT_0^c\cap \cE_\gamma} \\ &\le
     \sum_{j \in \cJ_1}\bP J j\bP{XY|J}{\cE_j\mid j} + \bP J {\cJ_2}
     \\ &\le \max_{j \in \cJ_1}\bP {XY|J}{\cE_j \mid j} + \frac 1N.
\end{align*}
Thus, the maximizing $j\in \cJ_1$ on the right satisfies the claimed
properties.  \qed

We now state our main converse bound.
\begin{theorem}[Single-shot converse] \label{t:converse_general}
For $0\le \ep <1$, $0< \eta < 1-\ep$, and parameters $\Delta, N$ as
above, the following lower bound on $L_\ep(X,Y)$ holds:
\begin{align*}
L_\ep(X,Y) &\ge \gamma + 3\log\left(\bPP \gamma- \ep - \bP{XY}{\cT_0} -
\frac 1N\right)_+ \\
&~~~ +\log(1-2\eta)-\Delta -6\log N - 4\log\frac 1
      {\eta}-1.
\end{align*}
\end{theorem} 
{\it Proof.} Let $j$ satisfy the properties stated in Lemma
\ref{l:good_index}.  The basic idea is to apply Theorem
\ref{t:converse_AU} to $\bPP{XY|\cE_j}$, where $\bPP{XY\mid \cE_j}$
denotes the conditional distributions on $X,Y$ given the event
$\cE_j$.

First, we have
\begin{align} \label{eq:lower-PXY}
\bP{XY|\cE_j}{x,y} \ge \bP{XY}{x,y}.
\end{align}
Furthermore, denoting $\alpha = \bPP \gamma - \ep -\bP{XY}{\cT_0} -
1/N$ and noting 
$\bP{J}{j} > 1/N^2$, we have for all $(x,y) \in \cE_j$ that
\begin{align}
\bP{XY|\cE_j}{x,y} &\le \frac {1}{\alpha}\bP {XY|J=j}{x,y} \\ &\le
\frac{N^2}{\alpha} \bP{XY}{x,y},
\label{eq:upper-PXY} 
\end{align}
where $\bPP{XY\mid J=j}$ denotes the conditional distributions on
$X,Y$ given $\{J=j\}$.  Thus, \eqref{eq:lower-PXY} and
\eqref{eq:upper-PXY} together imply, for all $(x,y)\in \cE_j$,
\[
\lambda_j + \log\alpha - 2 \log N \le -\log \bP{XY|\cE_j}{x,y} \le
\lambda_j+\Delta,
\]
$i.e.$, $\bPP{XY|\cE_j}$ is almost uniform with margin
$\Delta-\log \alpha + 2 \log N$
($cf.$~\eqref{e:uniformity_assumption}).  Also, note that
\eqref{eq:upper-PXY} implies
\begin{align*}
&\bP {XY|\cE_j}{-\log\frac{\bPP{XY|\cE_i}(X,Y)^2}{\bP X X\bP Y Y} \ge
    \gamma + 2\log \alpha -4\log N} 
\\ 
&\geq \bP {XY|\cE_j}{-\log\frac{\bP{XY}{X,Y}^2}{\bP X X\bP Y Y} \ge
    \gamma} 
\\
&= \bP {XY\mid 
    \cE_i}{\cE_\gamma} 
\\ 
&= 1,
\end{align*}
where the final equality holds by the definition of $\cE_\gamma$ in
\eqref{e:E_gamma_def}. Moreover, 
\[
\bP{XY|\cE_j}{X=\hat X, Y=\hat Y} = 1.  
\]
Thus, the proof
is completed by applying Theorem \ref{t:converse_AU} to
$\bPP{XY|\cE_j}$ with $\bQQ{X} = \bPP X$ and $\bQQ Y = \bPP Y$, and
$\Delta-\log \alpha + 2 \log N$ in place of $\Delta$.  \qed

\subsection{Converse bound for simple communication protocol} \label{subsec:converse-simple-protocol}
We close by noting a lower bound for the length of communication when
we restrict to simple communication. For simplicity assume that the joint
distribution $\bPP{XY}$ is indecomposable, $i.e.$, the 
{\it maximum common function} of $X$ and $Y$ is a constant (see
\cite{GacKor73}) and the parties can't agree on even a single bit
without communicating ($cf.$~\cite{Wit75}). The following bound
holds by a standard converse argument using the information spectral
method ($cf.$~\cite[Lemma 7.2.2]{Han03}).
\begin{proposition} \label{proposition:simple}
For $0 \le \ep < 1$, we have
\begin{align*}
 &L_\ep^{\mathrm{s}}(X,Y) \\
 &\ge \inf\bigg\{ l_1 + l_2 : \forall \delta >
 0, \\
 &~\mathbb{P}\Big( h(X|Y) > l_1 +\delta \mbox{ or } h(Y|X) > l_2 +
 \delta \Big) \le \ep + 2\cdot 2^{-\delta} \bigg\}.
\end{align*}
\end{proposition}
\begin{proof}
Since randomization (local or shared) does not help in improving the
length of communication ($cf.$~\cite[Chapter 3]{KushilevitzNisan97})
we can restrict to deterministic protocols. Then, 
since $\bPP{XY}$ is indecomposible, both 
parties have to predetermine the lengths of messages they send; let $l_1$ and $l_2$,
respectively, be the length of message sent by the first and the
second party. For $\delta > 0$, let
\begin{align*}
\cT_1 &:= \Big\{ (x,y): -\log \bP{X|Y}{x|y} \le l_1 + \delta \Big\},
\\ \cT_2 &:= \Big\{ (x,y): - \log \bP{Y|X}{y|x} \le l_2 + \delta
\Big\},
\end{align*}
and $\cT := \cT_1 \cap \cT_2$.  Let $\cA_1$ and $\cA_2$ be the set of
all $(x,y)$ such that party 2 and party 1 correctly recover $x$ and
$y$, respectively, and let $\cA := \cA_1 \cap \cA_2$.  Then, for any
simple communication protocol that attains $\ep$-DE, we have
\begin{align*}
\bP{XY}{\cT^c} &= \bP{XY}{\cT^c \cap \cA^c} + \bP{XY}{\cT^c \cap \cA}
\\ &\le \bP{XY}{\cA^c} + \bP{XY}{\cT_1^c \cap \cA} + \bP{XY}{\cT_2^c
  \cap \cA} \\ &\le \ep + \bP{XY}{\cT_1^c \cap \cA_1} +
\bP{XY}{\cT_2^c \cap \cA_2} \\ &\le \ep + 2 \cdot 2^{-\delta},
\end{align*}
where the last inequality follows by a standard argument
($cf.$~\cite[Lemma 7.2.2]{Han03}) as follows:
\begin{align*}
\bP{XY}{\cT_1^c \cap \cA_1} &\leq \sum_{y}\bP Y
y\bP{X|Y}{\cT_1^c\cap \cA_1| y}
\\
&\leq \sum_{y}\bP Y y |\{x: (x,y)\in \cA_1\}| 2^{-l_1 - \delta}
\\
&\leq \sum_{y}\bP Y y |\{x: (x,y)\in \cA_1\}| 2^{-l_1 - \delta}
\\
&\leq \sum_{y}\bP Y y 2^{ - \delta}
\\
&= 2^{-\delta},
\end{align*}
and similarly for $\bP{XY}{\cT_2^c \cap \cA_2}$; the desired bound follows.
\end{proof}

\section{General sources}\label{s:general_sources}
While the best rate of communication required for two parties to
exchange their data is known \cite{CsiNar04}, and it can be attained
by simple (noninteractive) Slepian-Wolf compression on both sides, the
problem remains unexplored for general sources. In fact, the answer is
completely different in general and simple Slepian-Wolf compression is suboptimal.

Formally, let $(X_n, Y_n)$ with joint distribution\footnote{The
  distributions $\bPP{X_nY_n}$ need not satisfy the consistency
  conditions.} $\bPP{X_nY_n}$ be a sequence of sources.  We need the
following concepts from the information spectrum method; see
\cite{Han03} for a detailed account. For random variables $(\bX,\bY) =
\{(X_n, Y_n)\}_{n=1}^\infty$, the the {\it inf entropy rate}
$\underline{H}(\bX \bY)$ and the {\it sup entropy rate}
$\overline{H}(\bX \bY)$ are defined as follows:
\begin{align*}
\underline{H}(\bX \bY) &= \sup\left\{\alpha \mid \lim_{n\rightarrow
  \infty} \bPr{\frac{1}{n}h(X_nY_n) < \alpha} = 0\right\},
\\ \overline{H}(\bX \bY) &= \inf\left\{\alpha \mid \lim_{n\rightarrow
  \infty} \bPr{\frac{1}{n}h(X_nY_n) > \alpha} = 0\right\};
\end{align*}
the {\it sup-conditional entropy rate} $\overline{H}(\bX| \bY)$ is
defined analogously by replacing $h(X_nY_n)$ with $h(X_n| Y_n)$. To
state our result, we also need another quantity defined by a
limit-superior in probability, namely the {\it sup sum conditional
  entropy rate}, given by
\begin{align*}
&\Romn \\ &= \inf\left\{\alpha \mid \lim_{n\rightarrow \infty}
  \bPr{\frac{1}{n} h(X_n \triangle Y_n)> \alpha} = 0\right\}.
\end{align*}

The result below characterizes $R^*(\bX, \bY)$ (see
Definition~\ref{d:R_star}).
\begin{theorem}\label{t:communication_omniscience_general}
For a sequence of sources $(\bX,\bY) = \{(X_n, Y_n)\}_{n=1}^\infty$,
\[
R^*(\bX,\bY) = \Romn.
\]
\end{theorem}
{\it Proof.} The claim follows from
Theorems~\ref{t:interactive_data_exchange}
and~\ref{t:converse_general} on choosing the spectrum slicing
parameters $\lamin, \lamax$, and $\Delta$ appropriately.

Specifically, using Theorem~\ref{t:interactive_data_exchange} with
\begin{align*}
\lamin &= n(\underline{H}(\bX, \bY) -\delta), \\ \lamax &=
n(\overline{H}(\bX, \bY) +\delta), \\ \Delta &= \sqrt{\lamax- \lamin}
\\&= N \\ \eta &= \Delta, \\ l_{\max} &= n (\Romn+\delta) + 3\Delta +1
\\ & = n(\Romn+\delta) + O(\sqrt{n}),
\end{align*}
where $\delta>0$ is arbitrary, we get a communication protocol of rate
$\Romn +\delta + O(n^{-1/2})$ attaining $\ep_n$-DE with $ \ep_n
\rightarrow 0$. Since $\delta > 0$ is arbitrary, $R^*(\bX,\bY) \leq
\Romn$.

For the other direction, given a sequence of protocols attaining
$\ep_n$-DE with $\ep_n \rightarrow 0$. Let
\begin{align*}
\lamin &= n(\underline{H}(\bX, \bY) -\Delta), \\ \lamax &=
n(\overline{H}(\bX, \bY) +\Delta),
\end{align*}
and so, $N = O(n)$. Using Theorem~\ref{t:converse_general} with
\[
\gamma = n(\Romn - \delta)
\]
for arbitrarily fixed $\delta > 0$, we get for $n$ sufficiently large
that
\begin{align*}
L_{\ep_n}(X_n, Y_n) &\geq n(\Romn - \delta) + o(n).
\end{align*}
Since $\delta > 0$ is arbitrary, the proof is complete.  \qed
\section{Strong converse and second-order asymptotics}
\label{s:strong_converse}

We now turn to the case of IID observations $(X^n, Y^n)$ and establish
the second-order asymptotic term in $L_{\ep}(X^n, Y^n)$.

\begin{theorem}\label{t:second_order}
For every $0< \ep < 1 $,
\begin{align*}
L_{\ep}\left(X^n, Y^n\right) = n H(X \triangle Y) + \sqrt{n V}
Q^{-1}(\ep) + o(\sqrt{n}).
\end{align*}
\end{theorem}
{\it Proof.} As before, we only need to choose appropriate parameters
in Theorems~\ref{t:interactive_data_exchange}
and~\ref{t:converse_general}.  Let $T$ denote the third central moment
of the random variable $\romn XY$.

For the achievability part, note that for IID random variables
$(X^n,Y^n)$ the spectrum of $P_{X^nY^n}$ has width
$O(\sqrt{n})$. Therefore, the parameters $\Delta$ and $N$ can be
$O(n^{1/4})$. Specifically, by standard measure concentration bounds
(for bounded random variables), for every $\delta>0$ there exists a
constant $c$ such that with $\lamax = nH(XY) + c\sqrt n$ and $\lamin =
nH(XY) -c\sqrt n$, 
\[
\bPr{(X^n, Y^n)\in \cT_0} \leq \delta.
\] 
For
\begin{align*}
\lambda_n &= n H(X \triangle Y) + \sqrt{n V} Q^{-1}\left(\ep - 2\delta
- \frac{T^3}{2V^{3/2}\sqrt{n}}\right),
\end{align*}
choosing $\Delta = N = \eta = \sqrt{2c}n^{1/4}$, and $l_{\max} =
\lambda_n + 3\Delta + 1$ in Theorem~\ref{t:interactive_data_exchange},
we get a protocol of length $\l_{\max}$ satisfying
\begin{align*}
\bPr{X \neq \hat X, \text{ or } Y \neq \hat Y} \leq \bPr{\sum_{i=1}^n
  h(X_i \triangle Y_i) > \lambda_n} + 2\delta,
\end{align*}
for $n$ sufficiently large.  Thus, the Berry-Ess\'een theorem
(\cf{}~\cite{Fel71}) and the observation above gives a protocol of
length $\l_{\max}$ attaining $\ep$-DE. Therefore, using the Taylor
approximation of $Q(\cdot)$ yields the achievability of the claimed
protocol length; we skip the details of this by-now-standard argument
(see, for instance, \cite{PolPooVer10}).

Similarly, the converse follows by Theorem~\ref{t:converse_general}
and the Berry-Ess\'een theorem upon choosing $\lamax$, $\lamin$, and
$N$ as in the proof of converse part of
Theorem~\ref{t:communication_omniscience_general} when $\lambda_n$ is
chosen to be
\begin{align*}
\lambda_n &= n H(X \triangle Y) + \sqrt{n V} Q^{-1}\left(\ep -
2\frac1N - \frac{T^3}{2V^{3/2}\sqrt{n}}\right) \\ &= n H(X \triangle
Y) + \sqrt{n V} Q^{-1}\left(\ep\right) + O(\log n),
\end{align*}
where the final equality is by the Taylor approximation of
$Q(\cdot)$.\qed

In the previous section, we saw that interaction is necessary to
attain the optimal first order asymptotic term in $L_\ep(X_n, Y_n)$
for a mixture of IID random variables. In fact, even for IID random
variables interaction is needed to attain the correct second order
asymptotic term in $L_\ep(X^n, Y^n)$, as shown by the following
example.

\begin{example}\label{ex:second_order_suboptimal}
Consider random variables $X$ and $Y$ with an indecomposable joint
distribution $\bPP{XY}$ such that the matrix
\begin{align*}
\mathbf{V} = \mathrm{Cov}([-\log \bP{X|Y}{X|Y}, -\log \bP{Y|X}{Y|X}])
\end{align*}
${\bf V}$ is nonsingular. For IID random variables $(X^n,Y^n)$ with
common distribution $\bPP{XY}$, using Proposition
\ref{proposition:simple} and a multidimensional Berry-Ess\'een theorem
($cf.$~\cite{TanKos14}), we get that the second-order asymptotic term
for the minimum length of simple communication  for 
$\ep$-DE is given by\footnote{The achievability part can be derived by
  a slight modification of the arguments in
  \cite{MiyKan95},\cite[Lemma 7.2.1]{Han03}.}
\begin{align*}
L_\ep^{\mathrm{s}}(X^n,Y^n) = n H(X \triangle Y)+ \sqrt{n} D_\ep +
o(\sqrt{n}),
\end{align*}
where
\begin{align*}
D_\ep := \inf\Big\{ r_1 + r_2 : \bPr{ Z_1\leq r_1, Z_2\leq r_2} \ge 1 - \ep \Big\},
\end{align*}
for Gaussian vector $\bZ = [Z_1,Z_2]$ with mean $[0,0]$ and covariance
matrix ${\bf V}$.  Since $\mathbf{V}$ is nonsingular,\footnote{For
  instance, when $X$ is uniform random variable on $\{0,1\}$ and $Y$
  is connected to $X$ via a binary symmetric channel, the covariance
  matrix $\mathbf{V}$ is singular and interaction does not help.}  we
have
\begin{align*}
\sqrt{V} Q^{-1}(\ep) &= \inf\Big\{ r: \mathbb{P}\Big(Z_1 + Z_2 \le r
\Big) \ge 1 - \ep \Big\} \\ &< D_\ep.
\end{align*}
Therefore, $L_\ep(X^n , Y^n)$ has strictly smaller second order term
than $L_\ep^s(X^n, Y^n)$, and interaction is necessary for attaining the optimal second
order term in $L_\ep(X^n, Y^n)$.
\end{example}

\section{Discussion}
We have presented an interactive data exchange protocol and a converse
bound which shows that, in a single-shot setup, the parties can
exchange data using roughly $h(X\Delta Y)$ bits when the parties
observe $X$ and $Y$. Our analysis is based on the information spectrum
approach. In particular, we extend this approach to
enable handling of interactive communication. A key step is the {\it spectrum
  slicing} technique which allows us to split a nonuniform distribution
into almost uniform ``spectrum slices''. Another distinguishing
feature of this work is our converse technique which is based on
extracting a secret key from the exchanged data and using an upper
bound for the rate of this secret key. In effect, this falls under the
broader umbrella of {\it common randomness decomposition} methodology
presented in \cite{TyaThesis} that studies a distributed computing problem by dividing the resulting common randomness into different independent components with operational significance.
 As a consequence, we obtain both the
optimal rate for data exchange for general sources as well as the
precise second-order asymptotic term for IID observations (which in
turn implies a strong converse). Interestingly, none of these optimal
results can be obtained by simple 
communication and interaction is necessary, in general. 
Note that our proposed scheme uses $O(n^{1/4})$ rounds of interaction; it remains open if fewer rounds of interaction will suffice.
 
Another asymptotic regime, which was not considered in this paper, is
the error exponent regime where we seek to characterize the largest
possible rate of exponential decay of error probability with
blocklength for IID observations. Specifically, denoting by
$\bP{\mathtt{err}}{l|X,Y}$ the least probability of error $\ep$ that
can be attained for data exchange by communicating less than $l$ bits,
$i.e.$,
\begin{align*}
\bP{\mathtt{err}}{l|X,Y} := \inf\{ \ep : L_\ep(X,Y) \le l \},
\end{align*}
we seek to characterize the limit of
\[
- \frac{1}{n} \log \bP{\mathtt{err}}{2^{nR}|X^n,Y^n}.
\]
The following result is obtained using a slight modification of our
single-shot protocol for data exchange where the slices of the
spectrum $\cT_i$ in \eqref{e:slice_i} are replaced with type classes and the decoder is replaced by a special case of the $\alpha$-decoder introduced in \cite{Csi82}.
For a fixed rate $R \geq 0$, our modified protocol enables data exchange, with
large probability, for every $(\bx, \by)$ with joint type
$\bPP{\overline{X}\,\overline{Y}}$ such that (roughly)
\[
R> H(\overline X\triangle \overline Y).
\]
The converse part follows from the strong converse of
Result~\ref{result:strong-converse}, together with a standard measure
change argument ($cf.$~\cite{CsiKor11}). The formal proof is given in Appendix A.

\begin{roughtheorem}[{\bf Error Exponent Behaviour}] \label{result:exponent}
For a given rate $R > H(X \triangle Y)$, define
\[
E_{\mathtt{r}}(R) := \min_{\bQQ{\overline{X}\,\overline{Y}}} \left[
  D(\bQQ{\overline{X}\,\overline{Y}} \| \bPP{XY}) + | R -
  H(\overline{X} \triangle \overline{Y})|^+ \right]
\]
and
\[
E_{\mathtt{sp}}(R) := \inf_{\bQQ{\overline{X}\,\overline{Y}} \in
  \cQ(R)} D(\bQQ{\overline{X}\,\overline{Y}} \| \bPP{XY}),
\]
where $|a|^+ = \max\{a,0\}$ and
\begin{align*}
\cQ(R) := \left\{ \bQQ{\overline{X}\,\overline{Y}} : R <
H(\overline{X} \triangle \overline{Y}) \right\}.
\end{align*}
Then, it holds that
\begin{align}
\liminf_{n\to\infty} - \frac{1}{n} \log
\bP{\mathtt{err}}{2^{nR}|X^n,Y^n} \ge E_{\mathtt{r}}(R) \nonumber
\end{align} 
and that
\begin{align}
\limsup_{n\to\infty} - \frac{1}{n} \log
\bP{\mathtt{err}}{2^{nR}|X^n,Y^n} &\le E_{\mathtt{sp}}(R).  \nonumber
\end{align}
\end{roughtheorem}
$E_{\mathtt{r}}(R)$ and $E_{\mathtt{sp}}(R)$, termed the {\it random
  coding exponent} and the {\it sphere-packing exponent}, may not
match in general. However, when $R$ is sufficiently close to $H(X
\triangle Y)$, the two bounds can be shown to coincide. In fact, in
Appendix B we exhibit an example where the optimal error exponent
attained by interactive protocols is strictly larger than that
attained by simple communication. Thus, in the error exponent regime,
too, interaction is strictly necessary.

\appendix

\subsection{Achievability Proof of Result \ref{result:exponent}}

In this appendix, we consider the error exponent and prove Result \ref{result:exponent}.
We use the method of types. The type
of a sequence $\mathbf{x}$ is denoted by $\bPP{\mathbf{x}}$. For a given type $\bPP{\overline{X}}$, the 
set of all sequences of type $\bPP{\overline{X}}$ is denoted by
$\cT_{\overline{X}}^n$. The set of all types on alphabet $\cX$ is
denoted by $\cP_n(\cX)$. We use similar notations for joint types and
conditional types. For a pair $(\mathbf{x},\mathbf{y})$ with joint
type $\bPP{\overline{X}\,\overline{Y}}$, we denote $H(\mathbf{x}
\triangle \mathbf{y}) = H(\overline{X} \triangle \overline{Y})$. We
refer the reader to~\cite{CsiKor11} for basic results on the method of
type.  

Fix $R > 0$ as the rate of communication exchanged (by both the
parties), but without adding the 
rate contributed by ACK-NACK messages exchanged. We consider $r$
rounds protocol, where $r = \lceil \frac{R}{\Delta} \rceil$ for
a fixed $\Delta > 0$. Let $R_i = i \Delta$ for $i =1 ,\ldots,r$. Basic
idea of the protocol is the same as our single-shot protocol, i.e., we
increment the hash size in steps. However, when we consider the error
exponent regime, to reduce the contribution of ``binning error'' to the error
exponent, we need a more carefully designed protocol.

For a given joint type
$\bPP{\overline{X}\,\overline{Y}}$, the key modification we make is to
delay the start of communication by Party 2 (which started once $R_i >
H(\overline{X} | \overline{Y})$ was satisfied). Heurisitcally, once Party 2 can decode $\mathbf{x}$ correctly, he can send $\mathbf{y}$ to Party 1 without error by using 
roughly\footnote{Since Party 2 has to send the joint type $\bPP{\overline{X}\,\overline{Y}}$ to Party 1,
additional $|\cX||\cY|\log(n+1)$ bits are needed.} 
$n H(\overline{Y}|\overline{X})$ bits, where $\bPP{\overline{X}\,\overline{Y}} = \bPP{\mathbf{x}\mathbf{y}}$. Thus, the budget Party 1 can use is $R - H(\overline{Y}|\overline{X})$,
which is larger than $H(\overline{X}|\overline{Y})$ when $R >
H(\overline{X} \triangle \overline{Y})$. Therefore, allowing Party 1
to communicate more before Party 2 starts may reduce the binning error probability.

Motivated by this reason, we assign the timing of decoding to each joint type as follows:
\begin{align*}
\phi(\bPP{\overline{X}\,\overline{Y}}) 
&:= \min\left\{ i : 1 \le i \le r, R_i \ge R - H(\overline{Y}|\overline{X}) - \Delta \right\} \\
&= \max\left\{ i : 1 \le i \le r, R_i <  R - H(\overline{Y}|\overline{X}) \right\}
\end{align*}
if $R - H(\overline{Y}|\overline{X}) - \Delta > 0$, and $\phi(\bPP{\overline{X}\,\overline{Y}}) = 0$ is $R - H(\overline{Y}|\overline{X}) - \Delta \le 0$.

For given hash functions $\mathbf{h} = (h_1,\ldots,h_r)$ with $h_i:\cX^n \to \{1,\ldots, 2^{\lceil n \Delta \rceil} \}$, 
let $N_{\mathbf{h}}(\overline{X}\,\hat{X}\,\overline{Y})$ denote, for each joint type $\bPP{\overline{X}\,\hat{X}\,\overline{Y}}$, the number of pairs 
$(\mathbf{x},\mathbf{y}) \in \cT_{\overline{X}\,\overline{Y}}^n$ such that for some 
$\hat{\mathbf{x}} \neq \mathbf{x}$ with $\bPP{\mathbf{x}\hat{\mathbf{x}}\mathbf{y}} = \bPP{\overline{X}\,\hat{X}\,\overline{Y}}$,
the relations 
\begin{align*}
h_i(\mathbf{x}) = h_i(\hat{\mathbf{x}}),~i =1,\ldots,\phi(\bPP{\hat{X}\,\overline{Y}})
\end{align*}
hold. The next result is a slight modification of a lemma in \cite[Section 3]{Csi82}; the proof
is almost the same and is omitted.
\begin{lemma} 
There exist hash functions $\mathbf{h} = (h_1,\ldots,h_r)$ such that 
for every joint type $\bPP{\overline{X}\,\hat{X}\,\overline{Y}}$ such
that $\phi(\bPP{\hat{X}\,\overline{Y}}) \neq 0$, the following bound holds:
\begin{align}
\frac{N_{\mathbf{h}}(\overline{X}\,\hat{X}\,\overline{Y})}{|\cT_{\overline{X}\,\overline{Y}}^n|} \le \exp\left\{ - n ( R_{\phi(\bPP{\hat{X}\,\overline{Y}})} - H(\hat{X} | \overline{X}\,\overline{Y}) - \delta_n) \right\},
\end{align}
where
\begin{align*}
\delta_n = |\cX|^2 |\cY| \frac{\log (n+1)}{n}.
\end{align*}
\end{lemma}
For the decoder, we use the {\em minimum sum conditional entropy
  decoder}, which is a kind of $\alpha$-decoder introduced in
\cite{Csi82}.

Our protocol is described in Protocol \ref{p:interactive-dx-exponent}.
\begin{protocol}[h]
\caption{Type-based interactive data exchange protocol}
\label{p:interactive-dx-exponent}
\KwIn{Observations $X^n$ and $Y^n$, parameter $\Delta$, and rate $R$}
\KwOut{Estimate $\hat X^n$ of $X^n$ at Party 2 and $\hat{Y}^n$ of $Y^n$ at Party 1}
\While{ $1\le i\le N$ and Party 2 did not send an ACK}
{
Party 1 sends $\prot_{2i-1} = h_i(X^n)$\\
\eIf{for $Y^n = \mathbf{y}$, Party 2 finds a unique $\mathbf{x}$ such that $\phi(\bPP{\mathbf{x}\mathbf{y}}) = i$,
$h_j(\mathbf{x}) = \prot_{2j-1},\, \forall\, 1\le j\le i$, and $H(\mathbf{x} \triangle \mathbf{y}) \le H(\hat{\mathbf{x}} \triangle \mathbf{y})$ for every $\hat{\mathbf{x}} \neq \mathbf{x}$
 satisfying $h_j(\hat{\mathbf{x}}) = \prot_{2j-1},\, \forall\, 1\le j\le i$} 
{
set $\hat X^n =  \mathbf{x}$\newline
send back $\prot_{2i} = \text{ACK}$ 
}
{
\eIf{More than one such $\mathbf{x}$ found}
{
protocol declares an error
}
{
send back $\prot_{2i} = \text{NACK}$
}
}
Reset $i \rightarrow i+1$
}
\If{No $\hat X^n$ found at Party 2}
{
Protocol declares an error
}
Party 2 send the joint type $\bPP{\overline{X}\,\overline{Y}}$ of $(\hat{X}^n,Y^n) = (\mathbf{x},\mathbf{y})$, and send the index of $\mathbf{y}$ among $\cT_{\overline{Y}|\overline{X}}^n(\mathbf{x})$.
\end{protocol}

The achievability part of Result \ref{result:exponent} can be seen as follows.
Fix a joint type $\bPP{\overline{X}\,\overline{Y}}$.
If $\phi(\bPP{\overline{X}\,\overline{Y}}) = 0$, then an error occurs
whenever $(\mathbf{x},\mathbf{y}) \in
\cT_{\overline{X}\,\overline{Y}}^n$. 
We also note that $R - H(\overline{Y}|\overline{X}) - \Delta \le 0$ implies $|R - H(\overline{X} \triangle \overline{Y}) - \Delta |^+ = 0$.
Thus, the probability of this kind of error is upper bounded by 
\begin{align} 
&\sum_{\bPP{\overline{X}\,\overline{Y}} \in \cP_n(\cX \times \cY) \atop  \phi(\bPP{\overline{X}\,\overline{Y}}) = 0} 
 \exp\{ - n [ D(\bPP{\overline{X}\,\overline{Y}} \| \bPP{XY}) \nonumber \\ 
&\hspace{35mm} + | R - H(\overline{X} \triangle \overline{Y})  - \Delta|^+ ]\}.
 \label{eq:error-first-kind}
\end{align}
Next, consider the case when  $\phi(\bPP{\overline{X}\,\overline{Y}}) \ge 1$.
For given conditional type $\bPP{\hat{X}|\overline{X}\,\overline{Y}}$, a sequence $\hat{\mathbf{x}}$ with
$(\mathbf{x},\hat{\mathbf{x}},\mathbf{y}) \in \cT^n_{\overline{X}\,\hat{X}\,\overline{Y}}$ causes an error when
\begin{enumerate}
\item $\phi(\bPP{\hat{X}\,\overline{Y}}) \le \phi(\bPP{\overline{X}\,\overline{Y}})$,

\item 
\begin{align*}
h_i(\hat{\mathbf{x}}) = h_i(\mathbf{x}),~i=1,\ldots,\phi(\bPP{\hat{X}\,\overline{Y}})
\end{align*}

\item 
\begin{align*}
H(\hat{X} \triangle \overline{Y})  \le H(\overline{X} \triangle \overline{Y}).
\end{align*}
\end{enumerate}
Also note that $\phi(P_{\overline{X}\,\overline{Y}}) = i$ implies 
\begin{align*}
H(\overline{Y}|\overline{X}) < R - R_i,
\end{align*}
i.e., once $\mathbf{x}$ is recovered correctly, $\mathbf{y}$ can be sent without an error.
Thus, the error probability of this kind is upper bounded by
\begin{align*}
& \sum_{ P_{\overline{X}\,\overline{Y} }\in \cP_n(\cX \times \cY) \atop  \phi(\bPP{\overline{X}\,\overline{Y}}) \ge 1 } 
 \exp\{ -nD(\bPP{\overline{X}\,\overline{Y} } \| \bPP{XY})\} \\
&~~~  \sum_{ \bPP{\hat{X}|\overline{X}\,\overline{Y}} \atop H(\hat{X} \triangle \overline{Y})  \le H(\overline{X} \triangle \overline{Y}) }
 \frac{N_{\mathbf{h}}(\overline{X}\,\hat{X}\,\overline{Y})}{|\cT_{\overline{X}\,\overline{Y}}^n|} \\
&\le \sum_{ \bPP{\overline{X}\,\overline{Y} } \in \cP_n(\cX \times \cY) \atop \phi(\bPP{\overline{X}\,\overline{Y}}) \ge 1}
 \exp\{ -nD(\bPP{\overline{X}\,\overline{Y} } \| \bPP{XY})\} 
 \sum_{ \bPP{\hat{X}|\overline{X}\,\overline{Y}} \atop H(\hat{X} \triangle \overline{Y})  \le H(\overline{X} \triangle \overline{Y})  } \\
&~~~ 
 \exp\{ - n |R_{\phi(\bPP{\hat{X}\overline{Y}})} - H(\hat{X}|\overline{X}\,\overline{Y}) - \delta_n|^+ \} \\
&\le \sum_{\bPP{\overline{X}\,\overline{Y} } \in \cP_n( \cX \times \cY) \atop \phi(\bPP{\overline{X}\,\overline{Y}}) \ge 1}
 \exp\{ -nD(\bPP{\overline{X}\,\overline{Y} } \| \bPP{XY})\} 
 \sum_{\bPP{\hat{X}|\overline{X}\,\overline{Y}} \atop H(\hat{X} \triangle \overline{Y})  \le H(\overline{X} \triangle \overline{Y}) } \\
&~~~ 
 \exp\{ - n |R - H(\overline{Y}|\hat{X}) - H(\hat{X}|\overline{X}\,\overline{Y}) - \Delta - \delta_n|^+ \} \\
&\le \sum_{\bPP{\overline{X}\,\overline{Y} } \in \cP_n(\cX \times \cY) \atop \phi(\bPP{\overline{X}\,\overline{Y}}) \ge 1}
 \exp\{ -nD(\bPP{\overline{X}\,\overline{Y} } \| \bPP{XY})\}
 \sum_{\bPP{\hat{X}|\overline{X}\,\overline{Y}} \atop H(\hat{X} \triangle \overline{Y})  \le H(\overline{X} \triangle \overline{Y}) } \\
&~~~ \exp\{ - n |R - H(\overline{Y}|\hat{X}) - H(\hat{X}|\overline{Y}) - \Delta - \delta_n|^+ \} \\
&= \sum_{\bPP{\overline{X}\,\overline{Y} } \in \cP_n(\cX \times \cY) \atop \phi(\bPP{\overline{X}\,\overline{Y}}) \ge 1}
 \exp\{ -nD(\bPP{\overline{X}\,\overline{Y} } \| \bPP{XY})\}
 \sum_{\bPP{\hat{X}|\overline{X}\,\overline{Y}} \atop H(\hat{X} \triangle \overline{Y})  \le H(\overline{X} \triangle \overline{Y}) } \\
&~~~ \exp\{ - n |R - H(\hat{X} \triangle \overline{Y}) - \Delta - \delta_n|^+ \} \\
&\le \sum_{\bPP{\overline{X}\,\overline{Y} } \in \cP_n(\cX \times \cY) \atop \phi(\bPP{\overline{X}\,\overline{Y}}) \ge 1}
 \exp\{ -nD(\bPP{\overline{X}\,\overline{Y} } \| \bPP{XY})\}
 \sum_{\bPP{\hat{X}|\overline{X}\,\overline{Y}} \atop H(\hat{X} \triangle \overline{Y})  \le H(\overline{X} \triangle \overline{Y}) } \\
&~~~ \exp\{ - n |R - H(\overline{X} \triangle \overline{Y}) - \Delta - \delta_n|^+ \} \\
&\le \sum_{ \bPP{\overline{X}\,\overline{Y} } \in \cP_n(\cX \times \cY) \atop \phi(\bPP{\overline{X}\,\overline{Y}}) \ge 1}
 \exp\{ -nD(\bPP{\overline{X}\,\overline{Y} } \| \bPP{XY})\}
 (n+1)^{|{\cal X}|^2  |{\cal Y}|} \\
&~~~ \exp\{ - n |R - H(\overline{X} \triangle \overline{Y}) - \Delta - \delta_n|^+ \}.
\end{align*}
Thus, by combining this with \eqref{eq:error-first-kind}, the total error probability is upper bounded by
\begin{align*}
&(n+1)^{|\cX|^2|\cY|+ |\cX||\cY|} 
 \exp\{ - n \min_{\bPP{\overline{X}\,\overline{Y}}} [ D(\bPP{\overline{X}\,\overline{Y} } \| \bPP{XY}) \\
 &~~~~~~~~~~~~~~~~~~~~~~~~~+  |R - H(\overline{X} \triangle \overline{Y}) - \Delta - \delta_n|^+ ] \}.
\end{align*}
Since $\Delta$ can be taken arbitrarily small, and the number of bits
needed to send ACK-NACK is at most $r$.\footnote{Our type-based
  protocol uses only constant number of rounds of interaction
  (independent of $n$).}
Consequently, Protocol \ref{p:interactive-dx-exponent} attains the exponent given in Result \ref{result:exponent}.

\subsection{An Example Such That Interaction Improves Error Exponent}

Consider the following source: $\cX$ and $\cY$ are both binary, and
$\bPP{XY}$ is given by
\begin{align*}
\bP{XY}{0,0} = \bP{XY}{1,0} = \bP{XY}{1,1} = \frac{1}{3},
\end{align*} 
that is, $X$ and $Y$ are connected by a $Z$-channel.  To evaluate
$E_{\mathtt{sp}}(R )$, without loss of generality, we can assume that
$\bQ{\overline{X}\,\overline{Y}}{0,1} = 0$ (since otherwise
$D(\bQQ{\overline{X}\,\overline{Y}} \| \bPP{XY}) = \infty$).  Let us
consider the following parametrization:
\begin{align*}
\bQ{\overline{X}\,\overline{Y}}{0,0} =
u,~\bQ{\overline{X}\,\overline{Y}}{1,0} = 1 - u -
v,~\bQ{\overline{X}\,\overline{Y}}{1,1} = v,
\end{align*}
where $0 \le u,v \le 1$.  Then, we have
\begin{align} \label{eq:divergence-u-v-form}
D(\bQQ{\overline{X}\,\overline{Y}} \| \bPP{XY}) = \log 3 -
H(\{u,1-u-v,v \})
\end{align}
and
\begin{align*}
 \lefteqn{ H(\overline{X}|\overline{Y}) + H(\overline{Y}|\overline{X}) } \\
 &=
 \kappa(u,v) \\ &:= (1-v) h\left(\frac{u}{1-v}\right) + (1-u)
 h\left(\frac{v}{1-u}\right).
\end{align*}
When the rate $R$ is sufficiently close to $H(X\triangle Y) =
\kappa(1/3,1/3) = 4/3$, the set $\cQ(R )$ is not empty.\footnote{In
  fact, we can check that $\frac{\kappa(z,z)}{dz}\Big|_{z=1/3} = -1$,
  and thus the function $\kappa(z,z)$ takes its maximum away from
  $(1/3,1/3)$.}  Since \eqref{eq:divergence-u-v-form} and
$\kappa(u,v)$ are both symmetric with respect to $u$ and $v$ and
\eqref{eq:divergence-u-v-form} and $\cQ(R )$ are convex function and
convex set, respectively, the optimal solution $(u^*,v^*)$ in the
infimum of $E_{\mathtt{sp}}(R )$ satisfies $u^*= v^*$.  Furthermore,
since $R > \kappa(1/3,1/3)$, we also have $u^* = v^* \neq 1/3$.

Note that for $R$ sufficiently close to $H(X\triangle Y)$,
$E_{\mathtt{sp}}(R )$ can be shown to equal $E_{\mathtt{r}}(R )$. Thus, to show that a
simple communication is strictly suboptimal for error exponent, it
suffices to show that $E_{\mathtt{sp}}(R ) >
E_{\mathtt{sp}}^{\mathrm{s}}(R )$, where the latter quantity
$E_{\mathtt{sp}}^{\mathrm{s}}(R )$ corresponds to the sphere packing
bound for error exponent using simple communication and is given by
\begin{align*}
E^{\mathrm{s}}_{\mathtt{sp}}(R ) := \max_{\substack{(R_1,R_2):
    \\ R_1+R_2 \le R}} \inf_{\bQQ{\overline{X}\,\overline{Y}} \in
  \cQ^{\mathrm{s}}(R_1,R_2)} D(\bQQ{\overline{X}\,\overline{Y}} \|
\bPP{XY})
\end{align*}
and
\begin{align*}
\cQ^{\mathrm{s}}(R_1,R_2) := \left\{ \bQQ{\overline{X}\,\overline{Y}}
: R_1 < H(\overline{X}|\overline{Y}) \mbox{ or } R_2 <
H(\overline{Y}|\overline{X}) \right\}.
\end{align*}
Since the source is symmetric with respect to $X$ and
$Y$, for evaluating $E_{\mathtt{sp}}^{\mathrm{s}}(R )$ we can assume 
without loss of generality that $R_1 \ge R_2$. Let
$u^\dagger := u^*$ and $v^\dagger := \frac{1-u^\dagger}{2}$ so that
$\frac{v^\dagger}{1-u^\dagger} = \frac{1}{2}$. Let
$\mathrm{Q}^{\dagger}_{\overline{X}\,\overline{Y}}$ be the
distribution that corresponds to $(u^\dagger,v^\dagger)$.  Note that
$\mathrm{Q}^{\dagger}_{\overline{X}\,\overline{Y}}$ satisfies
\begin{align*}
H(\overline{Y}|\overline{X}) &= (1-u^\dagger) h\left(
\frac{v^\dagger}{1-u^\dagger} \right) \\ &> (1-u^*) h\left(
\frac{v^*}{1-u^*} \right) \\ &\ge \frac{R}{2} \\ &\ge R_2,
\end{align*}
and so, $\mathrm{Q}^{\dagger}_{\overline{X}\,\overline{Y}} \in
\cQ^{\mathrm{s}}(R_1,R_2)$. For this choice of
$\mathrm{Q}^{\dagger}_{\overline{X}\,\overline{Y}}$, we have
\begin{align*}
D(\mathrm{Q}^{*}_{\overline{X}\,\overline{Y}} \| \bPP{XY}) &= \log 3 -
H(\{ u^*,1-u^* - v^*, v^*\}) \\ &= \log 3 - h(1-u^*) - (1-u^*) h\left(
\frac{v^*}{1-u^*}\right) \\ &> \log 3 - h(1-u^\dagger) - (1-u^\dagger)
h\left( \frac{v^\dagger}{1-u^\dagger} \right) \\ &=
D(\mathrm{Q}^{\dagger}_{\overline{X}\,\overline{Y}} \| \bPP{XY}),
\end{align*}
which implies $E_{\mathtt{sp}}(R ) > E_{\mathtt{sp}}^{\mathrm{s}}(R
)$.

\bibliography{IEEEabrv,references}
\bibliographystyle{IEEEtranS}


\end{document}